\documentclass[table,svgnames,dvipsnames]{article}

\usepackage{soul}
\usepackage{url}
\usepackage[hidelinks]{hyperref} 
\usepackage[small]{caption}
\usepackage{graphicx}
\usepackage{booktabs}
\usepackage{amsmath,multirow}
\usepackage{amsthm,thm-restate,thmtools}
\usepackage{amssymb}
\usepackage{enumerate}
\usepackage{xcolor,xfrac}
\usepackage[english]{babel}
\usepackage{charter}
\usepackage{authblk}
\usepackage{framed}
\usepackage[normalem]{ulem}  
\usepackage{tikz}
\usetikzlibrary{shapes.arrows}
\usetikzlibrary{shapes.geometric} 
\usepackage{amsfonts} 
\usepackage[T1]{fontenc}
\usepackage[utf8]{inputenc}
\usepackage[top=1 in,bottom=1in, left=1 in, right=1 in]{geometry}  
\usepackage{enumitem}
\usepackage{tabularx}
\hypersetup{
     colorlinks   = true,
     linkcolor    = red, % color of internal links
     urlcolor     = teal, % color of external links
   citecolor    = blue % color of links to bibliography
}
\usepackage{mathtools}
\usepackage{nicefrac} 
\usepackage{pifont}
\usepackage{bm}
\usepackage{dsfont}
\usepackage{algorithm}
\usepackage{graphicx}
\usepackage[noend]{algpseudocode}
\usepackage{subcaption}
\usepackage{gensymb}
\usepackage{pgfplots}
\usepackage{bbm}
\usepackage{comment}
\usepackage{xr}
\usepackage[sort,numbers]{natbib}

\usepackage{tkz-graph}

\usepackage[utf8]{inputenc}
\usepackage[T1]{fontenc}
\usepackage{amsmath,amssymb,amsthm,mathtools}
\usepackage{enumitem}
\usepackage{cleveref}

\theoremstyle{plain}
\newtheorem{theorem}{Theorem}[section]
\newtheorem{lemma}[theorem]{Lemma}
\newtheorem{proposition}[theorem]{Proposition}
\newtheorem{corollary}[theorem]{Corollary}

\theoremstyle{definition}
\newtheorem{definition}[theorem]{Definition}
\newtheorem{example}[theorem]{Example}

\theoremstyle{remark}
\newtheorem{remark}[theorem]{Remark}

\crefname{openproblem}{Open Problem}{Open Problems}
\Crefname{openproblem}{Open Problem}{Open Problems}

\newcounter{algctr}
\crefname{algctr}{Algorithm}{Algorithms}
\Crefname{algctr}{Algorithm}{Algorithms}

\newcommand{\N}{N}
\newcommand{\Rec}{R}
\newcommand{\charity}{C}
\newcommand{\val}{v}
\newcommand{\size}{s}
\newcommand{\bud}{B}
\newcommand{\dens}{\rho}
\newcommand{\densup}{\rho^{\max}}
\newcommand{\densdown}{\rho^{\min}}
\newcommand{\spread}{\Gamma}

\newcommand{\viol}{\mathrm{viol}}
\newcommand{\Alg}{\mathrm{Alg}}
\newcommand{\horizon}{T}
\DeclareMathOperator*{\argmin}{arg\,min}

\newcommand{\eps}{\varepsilon}

\providecommand{\Icom}{\mathcal I^{\mathrm{com}}_\gamma}

\providecommand{\horizon}{T}
\providecommand{\bud}{B}
\providecommand{\val}{v}
\providecommand{\size}{s}
\providecommand{\charity}{C}
\providecommand{\Rec}{R}
\providecommand{\N}{N}
\providecommand{\densup}{\rho^{\max}}
\providecommand{\densdown}{\rho^{\min}}
\makeatletter\@ifundefined{c@algctr}{\newcounter{algctr}}{}\makeatother

\allowdisplaybreaks

\usepackage{newfloat}
\usepackage{listings}
\DeclareCaptionStyle{ruled}{labelfont=normalfont,labelsep=colon,strut=off} % DO NOT CHANGE THIS
\floatstyle{ruled}
\newfloat{listing}{tb}{lst}{}
\floatname{listing}{Listing}

\usepackage{booktabs}

\title{Online Fair Division with Budget Constraints}
\author{Saar Cohen\textsuperscript{\rm *}, Nicholas Teh\textsuperscript{\rm *}, Paul W. Goldberg, Michael J. Wooldridge}
\affil[]{Department of Computer Science, University of Oxford, UK}
\date{\empty}

\begin{document}

\maketitle

\begin{abstract}
    We study an online variant of discrete fair division under generalized assignment budget constraints. Goods arrive one at a time and must be assigned irrevocably to a feasible agent or to charity, which holds all unallocated goods, while fairness is evaluated only against budget-feasible subsets of every recipient's bundle. We first show that, without additional structure, no deterministic online algorithm can guarantee any fixed approximation to feasible envy-freeness, even in highly symmetric instances. We then identify bounded density spread as a structural condition that restores meaningful guarantees, obtaining approximation algorithms for arbitrary item sizes and showing that, under common valuations and sufficiently small goods, these guarantees can be strengthened to an optimal deterministic frontier. We further study resource augmentation, where the online algorithm is allowed slightly larger budgets than the fairness benchmark, and characterize the resulting improvement in the achievable guarantees. Finally, we develop a learning-augmented framework based on predicting joint value-size types, proving consistency under perfect predictions, robustness to prediction error, and showing that separate predictions of value and size marginals are insufficient to recover strong fairness guarantees.
\end{abstract}

\section{Introduction}

\def\thefootnote{*}\footnotetext{Equal contribution.}

Fair division asks how resources should be allocated among agents with
heterogeneous preferences so that no agent feels unfairly treated.
In its classical form, all goods are available before the allocation is
computed, and every good is assigned to some agent. Many practical
allocation problems, however, must satisfy capacity or budget
constraints in addition to fairness. For instance, a food bank cannot
send a household more perishable goods than it can store, a cloud
scheduler cannot place more jobs on a machine than its available memory
permits, and a course-allocation mechanism must respect limits on the
number of credit hours assigned to each student.

Many of these applications are naturally modeled by fair division under generalized assignment constraints, introduced by \citet{barman2023guaranteeing}. In this framework, each good may have a different value and consume a different amount of capacity for each agent, and each agent may receive only a bundle whose total size does not exceed her budget. Since the agents' combined capacity may be insufficient to accommodate all goods, the standard convention is to introduce a distinguished recipient called \emph{charity}, which holds all unallocated goods~\citep{chaudhury2021little,wu2021budget,caragiannis2019envy,gan2023approximation}. This also changes the appropriate notion of fairness. Another recipient's entire bundle may not fit within an agent's budget, so comparing against that bundle is no longer meaningful. Instead, envy is evaluated only with respect to subsets that are feasible for the comparing agent. An allocation is therefore \emph{feasibly envy-free} if every agent weakly prefers her own bundle to every budget-feasible subset of every other recipient's bundle, including charity~\citep{barman2023guaranteeing}. For indivisible goods, this notion is naturally relaxed by allowing a bounded number of goods to be removed from the comparison set before evaluating envy.

A central assumption in this literature is that all goods, together
with the agents' values and sizes for them, are known \textit{before} the
allocation is computed. An algorithm can thus optimize over the
complete instance and coordinate its capacity decisions globally. In
many practical settings, however, this is unrealistic.
Goods arrive over time and must be committed \textit{immediately}. For example, donated items
are shelved as they arrive, jobs are scheduled when they are submitted,
and perishable resources cannot wait for a global optimization.

In this work, we therefore introduce the \textit{online} version of budget-constrained fair
division, in which indivisible goods arrive one at a time. When a
good arrives, the algorithm learns its complete value-size profile and
must irrevocably assign it either to an agent whose remaining budget can
accommodate it or to charity. The algorithm has no information about
future arrivals unless such information is explicitly supplied through
predictions. Fairness is evaluated with respect to all recipients,
including charity, and only against subsets that fit within the
comparing agent's original budget.

This problem is not merely the classical online fair-division problem
with an additional feasibility constraint. Budget constraints and
charity create two interacting difficulties that are absent from the
unconstrained model. First, accepting a good irrevocably consumes scarce
capacity, so an early assignment may prevent an agent from receiving a
more valuable good later. Second, a good that is not assigned to an
agent is placed in charity rather than disappearing from the instance.
Since fairness is evaluated against feasible subsets of the charity
bundle, rejecting a good may itself create future envy. Every online
decision therefore balances preserving capacity against avoiding future
envy.

These observations naturally raise the question of what additional
structure or algorithmic power is sufficient to recover meaningful
fairness guarantees in online budgeted fair division. Our contributions provide
several complementary answers to this question and are listed below.

\medskip
\textbf{Limits of online budgeted fair division.}
We prove that no deterministic online algorithm can guarantee any fixed
approximation to feasible envy-freeness in the unrestricted adversarial
model, even for two agents,
equal budgets, common valuations, uniform item sizes, and a known
horizon.

\medskip
\textbf{Bounded density spread restores tractability.}
For arbitrary item sizes, we derive capacity-aware online
algorithms whose fairness guarantees depend only on the variation in
value per unit of budget. Under common valuations, sizes common across
agents, and globally small goods, we further obtain matching upper and
lower bounds that exactly characterize the optimal deterministic
approximation frontier.

\medskip
\textbf{Value of additional online capacity.}
We study a \textit{resource-augmentation} model in which the online algorithm
may allocate using budgets that are slightly larger than the original
budgets against which fairness is evaluated. For arbitrary item sizes,
we quantify the resulting improvement in achievable fairness
guarantees and establish matching lower bounds. Under common valuations, sizes common across agents, and globally small
goods, we completely characterize the optimal tradeoff
between resource augmentation and fairness.

\medskip
\textbf{Learning-augmented online
budgeted fair division.}
Using predictions of future joint value-size types, we develop an online
algorithm that exactly reproduces any offline guarantee when predictions are
correct and degrades gracefully when they are not. We also prove that
this robustness guarantee is essentially tight and that predicting
values and sizes separately is fundamentally insufficient.

\section{Related Work}
\label{sec:related work}
%====================================================================

\textbf{Fair division under constraints and charity.}
Fair division under cardinality, matroid, connectivity, and budget
constraints has received considerable attention; see the survey
of~\citet{suksompong2021constraints}. Closest to our work is the
generalized assignment model of \citet{barman2023guaranteeing}, in which
goods have agent-specific values and sizes, agents are subject to
budgets, and unallocated goods are assigned to \emph{charity}. Building
on earlier work on constrained fair
division~\cite{chaudhury2021little,wu2021budget,caragiannis2019envy,gan2023approximation},
they formulate fairness through feasible envy, comparing each agent's
allocation only with budget-feasible subsets of other recipients'
bundles. These works study the \textit{offline} setting, whereas we introduce
its \textit{online} version, where goods arrive sequentially and allocation
decisions are irrevocable. \citet{Elkind2024} subsequently studied an analogous model for the setting with chores.

\medskip
\textbf{Online fair division.}
Online fair division was introduced through applications such as food
banks by~\citet{AleksandrovEtAl2015}; see the survey
of~\citet{AleksandrovWalsh2020Survey}. Subsequent work studies the
fairness-efficiency tradeoff under online
arrivals~\cite{BenadeEtAl2024,benade2025dynamic}, dynamic
reallocation~\cite{he2019achieving}, and online variants of envy-based
fairness notions~\cite{neoh2026online,choo2026approximate}. These works
often assume that every arriving good is allocated to an agent and
that fairness is evaluated with respect to entire bundles. 
Our model instead incorporates budget constraints and charity,
leading to capacity-aware feasible-envy comparisons; see
Appendix~\ref{supp:related} for a broader discussion.

\medskip
\textbf{EF1, EFX, and their feasible analogues.}
Envy-freeness up to one good (EF1)~\cite{budish2011combinatorial,lipton2004approximately}
and up to any good (EFX)~\cite{caragiannis2019unreasonable,neoh2025efx}
are standard relaxations of envy-freeness for indivisible goods.
Under budget constraints, however, another recipient's bundle may be
infeasible for the comparing agent. To address this issue,
\citet{barman2023guaranteeing} introduced feasible envy-freeness (FEF)
and its relaxation FEFx (building on related feasible-envy notions for other constraint classes), which restrict comparisons to budget-feasible
subsets and include charity as a recipient. We adopt these fairness
notions.

\medskip
\textbf{Learning-augmented online algorithms.}
Algorithms with machine-learned predictions aim to achieve
\emph{consistency} under accurate predictions while remaining
\emph{robust} to prediction error~\cite{purohit2018improving,lykouris2021competitive};
see the survey of~\citet{Mitzenmacher_Vassilvitskii_2021}. Learning-augmented
techniques have recently been applied to online fair
division~\cite{banerjee2022online,barman2022universal,cohen2023general,an2024best,huang2025long,zhou2023multi,neoh2026online,spaeh2023online,balkanski2023strategyproof,cohen2024plant,choo2026approximate,wang2026online}.
Unlike these works, which predict aggregate values,
normalization information, value frequencies, or maximum
item values, we predict joint value-size types,
reflecting that budget feasibility depends on value-size
correlations; see Appendix~\ref{supp:related} for further
discussion.

\medskip
\textbf{Budgeted allocation and knapsack.}
Our model shares the capacity-constrained allocation structure of
generalized assignment, online knapsack, and AdWords-style
problems~\cite{martello1990knapsack,mehta2007adwords}. Unlike these
problems, however, our objective is fairness rather than welfare or
revenue. The logarithmic frontier that we obtain in the small-item
setting is closely related to the optimal competitive ratio for online
knapsack with bounded value-to-size
ratios~\cite{zhou2008budget}, although our multi-agent fairness analysis
requires additional techniques.
%====================================================================
%====================================================================
\section{Preliminaries}
\label{sec:model}
%====================================================================

For a positive integer $z$, let $[z]=\{1,\dots,z\}$. We study an online allocation problem in which a fixed set of agents receives indivisible goods that arrive one at a time, and each arrival must be assigned irrevocably before the next good is revealed.
Formally, there is a fixed set of agents $\N=[n]$. Each agent
$i\in\N$ has a strictly positive budget $\bud_i$, which limits the
total size of the goods she may receive.
Goods arrive sequentially in discrete rounds. Let
$G=\{g_1,\dots,g_T\}$ denote the finite set of goods, where
$g_t$ is the good arriving in round $t$. By a slight abuse of
notation, we also use $G$ to refer to this arrival sequence
when the order is relevant. For each round $t\in[\horizon]$, we denote by
$G^t=\{g_1,\dots,g_t\}$ the set of goods that have arrived by the end
of round $t$. Goods arriving in different rounds are treated as
distinct, even if they have identical values and sizes.

Each agent $i\in N$ has an additive valuation
$v_i:2^G\to\mathbb R_{\ge 0}$ and an additive size function
$s_i:2^G\to\mathbb R_{\ge 0}$. For every good $g\in G$,
we assume $v_i(g)\ge 0$ and $s_i(g)>0$.
For every bundle $S\subseteq G$, $v_i(S)=\sum_{g\in S}v_i(g)$ and $s_i(S)=\sum_{g\in S}s_i(g)$.
A bundle $S$ is feasible for agent $i$ if $s_i(S)\le B_i$.
A good $g$ is individually feasible for agent $i$ if $s_i(g)\le B_i$.
An instance of our problem is thus specified by $I=(\N,G,(\bud_i,\val_i,\size_i)_{i\in\N})$.

Budget constraints may make it impossible, or undesirable, to assign
every good to an agent.  We therefore include a distinguished
recipient $\charity$, called \emph{charity}, which holds all goods
that are not assigned to agents.  Let
$\Rec=\N\cup\{\charity\}$ denote the set of recipients.

An allocation after round $t$ is a tuple
$A^t=(A_r^t)_{r\in\Rec}$ that partitions $G^t$.  It is feasible
if $\size_i(A_i^t)\le\bud_i$ for every agent $i$; charity is not
subject to a budget constraint. We write $A=A^{\horizon}$ for the
final allocation and abbreviate its components by $A_i$ and
$A_{\charity}$.

The interaction is online and irrevocable.
After observing $g_t$, the algorithm must immediately choose one
recipient $r_t\in\Rec$. Assigning $g_t$ to agent $i$ is
permitted only if $\size_i(A_i^{t-1})+\size_i(g_t)\le\bud_i$,
while assignment to charity is always permitted.  Once a good has
been assigned, its recipient cannot change. 

We distinguish between unknown and known horizons.  In the
\emph{unknown-horizon} model, the number of arrivals is not revealed
to the algorithm, and a guarantee must remain valid whenever the
finite sequence stops.  In the \emph{known-horizon} model, the total
number of rounds is announced before the first arrival.
We also use the standard distinction between adaptive and oblivious
adversaries. An \textit{adaptive adversary} may choose later arrivals after
observing the algorithm's earlier decisions, while an \textit{oblivious
adversary} fixes the entire input sequence before its execution.
Unless explicitly stated otherwise, algorithms in this
paper are \textit{deterministic}.

%--------------------------------------------------------------------
\subsection{Valuation, Size, and Density Structure}
\label{sec:structural-classes}
%--------------------------------------------------------------------

Our general model allows values and sizes to depend on both the agent
and the good.  Several of our results impose additional structure.
We collect the recurring special cases here so that later sections
can state their assumptions compactly.
Valuations are \emph{common} if there is a single function
$\val:G\to\mathbb R_{\ge 0}$ such that
$\val_i(g)=\val(g)$ for every $i\in\N$ and $g\in G$.
They are \emph{scaled common} if there is a base value function
$\val$ and positive constants $\beta_i$ satisfying
$\val_i(g)=\beta_i\val(g)$ for every agent and good.  Common
valuations are the special case in which all scales are equal.

The size assumptions require a separate distinction.  Saying that
sizes are common across agents does not mean that all goods have the
same size.

\begin{definition}[Sizes common across agents and uniform item sizes]
\label{def:common-uniform-sizes}
Sizes are \emph{common across agents} if there is a function
$s:G\to\mathbb R_{>0}$ such that $\size_i(g)=s(g)$ for every agent $i\in\N$ and good $g\in G$.
Thus, different goods may still have different sizes.
Sizes are \emph{uniform across goods} if there is a constant
$\delta>0$ such that $\size_i(g)=\delta$ for every agent $i\in\N$ and good $g\in G$.
Uniform item sizes therefore imply sizes common across agents, but
the converse need not hold.
\end{definition}
For $\sigma_{\max}\ge0$, an instance is globally
$\sigma_{\max}$-small if $s_i(g)\le \sigma_{\max}B_i$ for every $i \in N$ and $g \in G$.
When budgets are $1$ and sizes are common across agents, this
condition reduces to $s(g)\le\sigma_{\max}$ for every $g \in G$.

Whenever $\val_i(g)>0$, define the density
$\dens_i(g)=\val_i(g)/\size_i(g)$, the value agent $i$
derives per unit of budget consumed. Thus, higher-density
goods are more budget-efficient. Many of our guarantees
depend not on absolute densities but on their variation
within an instance, captured by the ratio of the largest and
smallest positive densities.

\begin{definition}[Density spread]
\label{def:density-spread}
For an agent $i$ having at least one good $g$ such that
$v_i(g)>0$ and $s_i(g)\le B_i$, define $\rho_i^{\min} := \min\{ \rho_i(g): g\in G, v_i(g)>0, s_i(g)\le B_i \}$
and $\rho_i^{\max} := \max\{ \rho_i(g): g\in G, v_i(g)>0, s_i(g)\le B_i\}$.
The density spread of agent $i$ is $\Gamma_i
:=
\frac{\rho_i^{\max}}{\rho_i^{\min}}
\ge 1$.
If no good satisfies $v_i(g)>0$ and $s_i(g)\le B_i$, set
$\Gamma_i:=1$; in this case, every set feasible for $i$ has
value zero for $i$.
The density spread of the instance is $\Gamma:=\max_{i\in N}\Gamma_i$.
\end{definition}

Only positively valued goods that are individually feasible
enter $\Gamma_i$, since other goods cannot belong to any
feasible comparison set for agent $i$.

%--------------------------------------------------------------------
\subsection{Feasible Envy-Based Fairness Notions}
\label{sec:fairness-notions}
%--------------------------------------------------------------------

Under generalized assignment constraints, another recipient's bundle
may be infeasible for the comparing agent. Following
\citet{barman2023guaranteeing} and related work on constrained fair
division~\cite{chaudhury2021little,wu2021budget,caragiannis2019envy,gan2023approximation},
we therefore compare an agent's bundle only with budget-feasible subsets
of every recipient's bundle, including charity. The strongest notion
requires the agent to weakly prefer her own bundle to every such
comparison set.

\begin{definition}[Feasible envy-freeness]
\label{def:fef}
An allocation $A$ is \emph{feasibly envy-free} (FEF) if $\val_i(A_i)\ge \val_i(S)$ for any agent $i\in\N$, recipient
$r\in\Rec$, and set $S\subseteq A_r$ feasible for $i$.
\end{definition}

Exact FEF is generally too demanding for indivisible goods under
generalized assignment constraints, as feasible envy-free allocations
need not exist~\cite{barman2023guaranteeing}. We therefore consider
relaxations that permit the comparison set to lose a bounded number of
goods. Our primary notion is a multiplicative approximation of
FEF$ k$, which generalizes feasible envy-freeness up to one good
and related notions studied in the constrained fair-division
literature~\cite{barman2023guaranteeing,wu2021budget,chaudhury2021little}.

\begin{definition}[Multiplicative $\alpha$-FEF$ k$]
\label{def:afefk}
Let $\alpha\in[0,1]$ and let $k\ge 0$ be an integer.  An allocation
$A$ is \emph{$\alpha$-FEF$ k$} if, for every agent
$i\in\N$, recipient $r\in\Rec$, and set $S\subseteq A_r$
feasible for $i$, there exists a set $X\subseteq S$ with
$|X|\le k$ such that $\val_i(A_i) \ge \alpha\,\val_i(S\setminus X)$.
When $\alpha=1$, we simply write FEF$ k$.
\end{definition}
For an online algorithm, an $\alpha$-FEF$k$ guarantee is
\emph{prefix-wise} if, after every round $t\in[T]$, the
current allocation $A^t$ of $G^t$ is $\alpha$-FEF$k$.
The most prominent case is FEF$1$, which permits the removal of at
most one good from the comparison set.  A related, stronger notion
requires the inequality to survive the removal of \emph{every} good.

\begin{definition}[FEFx]
\label{def:fefx}
An allocation $A$ is \emph{feasibly envy-free up to any good}
(FEFx) if, for every agent $i\in N$, every recipient $r\in R$,
and every strict subset $S\subsetneq A_r$ that is feasible for $i$, $v_i(A_i)\ge v_i(S)$.
\end{definition}

Multiplicative guarantees are natural when values may be rescaled.
For the prediction-based results later in the paper, it is also useful to
measure fairness loss additively.

\begin{definition}[Additive $\eta$-FEF$ k$]
\label{def:etafefk}
Let $\eta\ge 0$ and let $k\ge 0$ be an integer.  An allocation
$A$ is \emph{additive $\eta$-FEF$ k$} if, for any agent
$i\in\N$, recipient $r\in\Rec$, and set $S\subseteq A_r$
feasible for $i$, there exists $X\subseteq S$ with
$|X|\le k$ such that $\val_i(A_i) \ge \val_i(S\setminus X)-\eta$.
\end{definition}

The notions presented above satisfy the natural hierarchy inherited
from their classical counterparts. In particular,
FEF$0$ implies FEF$1$, which in turn implies
FEF$2$, and more generally FEF$ k$ implies
FEF${k'}$ whenever $k\le k'$. Likewise, FEFx implies
FEF$1$ by definition. Finally, if every feasible comparison set has
cardinality at most $k$, then every feasible allocation is
FEF$ k$, as the entire comparison set may be removed.
This explains why allowing $k$ to depend freely on
the instance would make the notion uninformative.  

\begin{remark}[$k$ must be instance-independent]
\label{remark:kappa-trivial}
For an instance $I$, let its \emph{maximum feasible comparison
cardinality} be $\kappa(I) := \max_{i\in\N} \max\{ |S|: S\subseteq G,\size_i(S)\le\bud_i\}$. Since every feasible comparison set contains at most
$\kappa(I)$ goods, every feasible allocation is trivially
FEF${\kappa(I)}$: one may simply remove the entire
comparison set. Thus, meaningful FEF$ k$ guarantees require $k$ to remain independent of the instance. 
\end{remark}

%====================================================================
\section{No Fixed $\alpha$-FEF$k$ Guarantee in General}\label{sec:lb}
%====================================================================

We first show that no deterministic algorithm can guarantee any fixed
multiplicative $\alpha$-FEF$k$ relaxation, even against an
oblivious adversary and regardless of the values of $\alpha$ and
$k$. Thus, without additional structural assumptions, even very
weak relaxations of feasible envy-freeness are unattainable online.
Perhaps surprisingly, this impossibility already holds on highly
symmetric instances. The hard instances involve only two agents with
equal budgets, common valuations, uniform item sizes, and even a known
horizon. The only source of difficulty is the interaction between
irrevocable budget consumption and charity-inclusive feasible envy.

The proof uses phases whose values grow sufficiently quickly that
each phase dominates the total value accumulated in all previous
phases. To avoid an immediate violation of $\alpha$-FEF$k$, the
algorithm must devote a positive amount of capacity to every phase.
Since capacities are finite, these requirements eventually become
incompatible. Although the proof is conveniently described through
a phase simulation, determinism allows the resulting finite hard
sequence to be fixed before its actual execution.

\begin{restatable}{theorem}{MainImpossibility}
\label{thm:main-lb}
Fix any integer $k\ge 0$ and any $\alpha\in(0,1]$.
For every deterministic online algorithm in the
\emph{irrevocable} model, there exists a fixed finite input
sequence on which the final allocation is not
$\alpha$-FEF$k$. Equivalently, no deterministic online
algorithm guarantees $\alpha$-FEF$k$ even against an
\emph{oblivious} adversary. The impossibility holds even
when simultaneously: \textup{(1)} there are only two agents, $n=2$; \textup{(2)} the budgets are equal, $\bud_1=\bud_2=1$; \textup{(3)} valuations are common; \textup{(4)} sizes are uniform across goods; \textup{(5)} the horizon $\horizon$ is known to the algorithm in advance.
\end{restatable}

The proof appears in Appendix~\ref{app:proof:thm:main-lb}.
\Cref{thm:main-lb} immediately yields two useful
specializations. Setting $\alpha=1$ shows that exact
FEF$ k$ is impossible for every fixed value of $k$.

\begin{corollary}
\label{cor:exact-fefk}
For every fixed integer
$k\ge 0$, no deterministic online algorithm guarantees exact
FEF$k$ on all instances, even against an \emph{oblivious}
adversary and even under assumptions
\textup{(1)--(5)} of \Cref{thm:main-lb}.
\end{corollary}

The hard instances constructed in the proof of
\Cref{thm:main-lb} have an additional structural
property. Since every good has size $1/Q$ and every budget
is $1$, every feasible comparison set contains at most
$Q$ goods. Consequently, the same construction yields an
impossibility result parameterized by the maximum feasible
comparison cardinality. We stress that the result establishes
the existence of a sufficiently large value $q>k$; it does
\emph{not} claim that hardness holds for every $q>k$. The proof appears in Appendix~\ref{app:proof:cor:kappa}.

\begin{restatable}{corollary}{CardinalityImpossibility}
\label{cor:kappa}
For every fixed integer
$k\ge 0$, there exists an integer $q>k$ such that no
deterministic online algorithm guarantees exact FEF$k$ on all
instances satisfying $\kappa(I)=q$ (Remark~\ref{remark:kappa-trivial}), even against an
\emph{oblivious} adversary and even under assumptions
\textup{(1)--(5)} of \Cref{thm:main-lb}.
\end{restatable}

\Cref{thm:main-lb} is not implied by existing impossibility
results for online EF1 by \citet{neoh2026online}, which assume that every good is
allocated to an agent and measure envy with respect to
entire bundles. Our lower bound instead relies on the
interaction between budget constraints and charity-inclusive
feasible envy: accepting goods irrevocably consumes
capacity, while goods sent to charity remain valid objects
of comparison. The hardness therefore arises from
capacity-limited feasible comparisons rather than from
heterogeneous preferences or unknown horizons.

The construction of \Cref{thm:main-lb} has unbounded density
spread. Thus, the theorem establishes impossibility in the
absence of additional structure, but does not determine the
best guarantee under bounded density spread. Appendix~\ref{sec:density}
addresses this question by identifying the optimal limiting
frontier in the globally small common-valuation setting
through a separate lower-bound construction.
%====================================================================
\section{Greedy Algorithms under Bounded Density Spread}
\label{sec:greedy}

We now identify bounded density spread as a structural condition under
which meaningful online fairness guarantees become possible. For
arbitrary item sizes, we derive simple capacity-aware greedy algorithms
that achieve multiplicative guarantees under common, scaled-common, and
fully heterogeneous valuations. Appendix~\ref{sec:density} 
strengthens these results in the globally-small common-valuation model,
where a different algorithm yields a sharp deterministic frontier.

\subsection{Common Valuations}
\label{sec:Common valuations}

We begin with the common-valuation setting, where every agent agrees on
the value of each good. In this setting, we introduce
\textsc{Positive-GreedyFit} (\Cref{alg:pgf}), a deterministic online greedy algorithm that
forms the basis of all positive results in this section. The algorithm balances two
competing objectives. First, it never allocates a good of
zero value to an agent (line \ref{state:charity}), thereby preserving scarce budget capacity for
future arrivals. Second, whenever a positive-value good is
feasible for multiple agents, it assigns the good to an agent whose
current bundle has minimum value (lines \ref{state:feasible agents}--\ref{state:min}), thus balancing the accumulated
value across agents.  The algorithm is deterministic, processes each arrival in
$\mathcal{O}(n)$ time, and maintains feasibility by construction,
since goods are assigned only to agents with sufficient
remaining budget.

\begin{algorithm}[t]
\caption{\textsc{Positive-GreedyFit} (common valuations)}
\label{alg:pgf}
\begin{algorithmic}[1]
\State Initialize $A_i\gets\varnothing$ for all $i\in\N$ and
      $A_{\charity}\gets\varnothing$
\For{each arriving good $g$}
    \If{$\val(g)=0$} {assign $g$ to charity \label{state:charity}}
    \Else
        { Let $
        F_g
        \gets
        \{
        i\in\N :
        \size_i(A_i)+\size_i(g)\le\bud_i
        \}
        $ be the agents that can accommodate good $g$.\label{state:feasible agents}}
        \If{$F_g=\varnothing$} {assign $g$ to charity.}
        \Else{ assign $g$ to an arbitrary
            $i^*\in\argmin_{h\in F_g}\val(A_h)$,
            using a fixed tie-breaking order. \label{state:min}}
        \EndIf
    \EndIf
\EndFor
\State \Return
$
A=(A_1,\dots,A_n,A_{\charity})
$
\end{algorithmic}
\end{algorithm}

The following theorem shows that bounded density spread is sufficient to
recover a multiplicative FEF$1$ guarantee. Moreover, the guarantee
depends only on the density spread.

\begin{restatable}{theorem}{DensitySpreadGuarantee}
\label{thm:greedy-common}
\textsc{Positive-GreedyFit} guarantees
$\spread^{-1}$-FEF$1$ under common valuations. In particular, if
$\spread=1$, then the algorithm returns an exact
FEF$1$ allocation.
\end{restatable}

\emph{Proof sketch.}
In Appendix~\ref{app:proof:thm:greedy-common}, we fix a feasible comparison set and remove its last-arriving
positively valued good $f$. If $f$ was feasible for the
comparing agent when it arrived, the minimum-value rule
bounds the remaining comparison value. Otherwise,
infeasibility yields a size inequality, which translates to
the factor $\Gamma_i^{-1}$ via
$\rho_i^{\min}$ and $\rho_i^{\max}$.

\Cref{thm:greedy-common} is consistent with the
impossibility result of \Cref{thm:main-lb}. In the
family of instances used in the proof of
\Cref{thm:main-lb}, all goods have the same size but the phase
values increase rapidly, so the density spread of
that family can be arbitrarily large and the factor
$\Gamma^{-1}$ can be arbitrarily small. Thus, \Cref{thm:greedy-common} gives a nonzero instance-independent factor only when $\Gamma$ is
bounded in advance.

\textbf{Scaled-common valuations.}
The analysis of \textsc{Positive-GreedyFit} extends directly to
scaled-common valuations, where
$\val_i(g)=\beta_i\val(g)$ for a common base valuation $\val$.
Running the algorithm on $\val$ preserves both the ordering of goods and
the density spread, so the $\Gamma^{-1}$-FEF1 guarantee of
\Cref{thm:greedy-common} continues to hold. The formal statement and proof are
given in Appendix~\ref{app:proof:thm:scaled}.

\subsection{Fully Heterogeneous Valuations}
\label{sec:Fully heterogeneous valuations}

We now allow agent-specific additive valuations. The proof from
the common-valuation case cannot be used directly because agent
$i$ and the recipient of a good may assign different values to
that good.
The algorithm below compares the numerical quantities $v_i(A_i)$
across agents. We therefore assume in this subsection that the
reported numerical values are directly comparable across agents.
Rescaling only one agent's valuation may change both the
allocation produced by the algorithm and the bound below.

The algorithm assigns a good only to feasible agents that value it
positively. To compare valuations across agents, define
$\chi_i:=\max(\{1\}\cup\{v_i(g)/v_j(g):v_j(g)>0\})$,
which is a \emph{cross-comparability} bound that measures
how much larger agent $i$'s value for a good can be than
that assigned by any recipient who values it positively.
Now, define
$\Lambda_i:=\max\{\Gamma_i,\chi_i\}$ and
$\Lambda:=\max_{i\in\N}\Lambda_i$, where $\Gamma_i$
captures the loss due to capacity constraints and $\chi_i$
the additional loss from heterogeneous valuations. The
analysis incurs the former when the last relevant good is
infeasible for agent $i$, and the latter otherwise, yielding
a $\Lambda_i^{-1}$ guarantee.
\textsc{Hetero-GreedyFit} differs from
\textsc{Positive-GreedyFit} only by restricting eligible
agents to those assigning positive value to the arriving
good; see Appendix~\ref{app:hgf}. The following theorem
establishes the corresponding $\Lambda^{-1}$-FEF$1$
guarantee.

\begin{restatable}{theorem}{HeterogeneousGuarantee}
\label{thm:hetero}
\textsc{Hetero-GreedyFit}
guarantees $\Lambda^{-1}$-FEF$1$ under agent-specific additive valuations whose numerical values
are directly comparable across agents. In particular, if both
the density spread and the cross-comparability are bounded,
i.e., $\Lambda\le\lambda$ for some constant
$\lambda$, the algorithm guarantees
$\lambda^{-1}$-FEF$1$.
\end{restatable}

The proof is similar to that of \Cref{thm:greedy-common} and thus deferred to Appendix~\ref{app:proof:thm:hetero}. The
guarantee of \Cref{thm:hetero} is meaningful when $\chi_i$ is bounded.
If the base valuation is not identically zero, scaled-common
valuations satisfy
$v_i(g)/v_j(g)=\beta_i/\beta_j$ on every positively valued good, so
$\chi_i=\max_j\beta_i/\beta_j$; otherwise, all fairness comparisons are
trivial. More generally, any valuation class with uniformly bounded
$v_i(g)/v_j(g)$ for positively valued goods yields a constant-factor
FEF$1$ guarantee. On scaled-common valuations, however,
\Cref{thm:scaled} remains sharper, as it exploits the common base
ranking and avoids the additional factor $\chi_i$. All three results,
\Cref{thm:greedy-common,thm:hetero,thm:scaled}, allow arbitrary good
sizes. Appendix~\ref{sec:sharp-frontier} strengthens the common-valuation
analysis for common sizes, unit budgets, and goods whose maximum size
tends to zero by determining the exact limiting factor, while
Appendix~\ref{sec:aux-density} gives a separate known-horizon lower
bound for arbitrary good sizes.

%======================================================================
\section{Resource Augmentation for Arbitrary Sizes}\label{sec:aug}
%======================================================================

The previous section showed that bounded density spread restores
meaningful multiplicative guarantees, but the approximation factor
deteriorates as the density spread increases. Moreover,
\Cref{thm:main-lb} rules out any fixed feasible-envy relaxation in
complete generality. A standard way to overcome such online
limitations is \emph{resource augmentation}
\citep{kalyanasundaram2000speed,roughgarden2021beyond}, in which the
algorithm is given more resources than the benchmark against which it
is evaluated. Since the difficulty in our setting stems from
irrevocable budget consumption, we augment the agents' online budgets.
Fairness is still evaluated with respect to the original budgets, so
only the algorithm's capacity constraints are relaxed. The results
below apply to arbitrary item sizes. Appendix~\ref{sec:density} shows
that under the additional globally-small assumption, resource
augmentation admits a substantially sharper characterization.

Formally, fix an augmentation parameter $\eps\ge0$. Each agent $i$
has original budget $\bud_i$, while the online algorithm
allocates using the enlarged \emph{algorithmic budget}
$(1+\eps)\bud_i$. Fairness, however, is still evaluated
with respect to the original budgets. We thus call an
allocation $A=(A_1,\dots,A_n,A_{\charity})$
\emph{bicriteria-feasible} if
$\size_i(A_i)\le(1+\eps)\bud_i$ for each agent $i$.
Throughout this section, every $\alpha$-FEF$k$ guarantee is
interpreted using the original budgets, i.e., any comparison
set must satisfy $\size_i(S)\le\bud_i$. Hence, resource
augmentation relaxes only the algorithm's capacity
constraints, not the fairness benchmark. When $\eps=0$, we
recover the original model.

We introduce
\textsc{Aug-GreedyFit}$(\eps)$, which modifies
\textsc{Positive-GreedyFit} by replacing each budget $\bud_i$
with $(1+\eps)\bud_i$. A good is
eligible for agent $i$ only if it is individually feasible
under the original budget $\bud_i$ and its assignment keeps
the total allocated size within $(1+\eps)\bud_i$. The first
condition excludes goods that can never appear in a feasible
comparison set.
See Appendix~\ref{app:aug-greedyfit} for a
pseudocode.

%----------------------------------------------------------------------
Resource augmentation enlarges the algorithm's capacity,
improving the achievable fairness guarantee. The following
result quantifies this tradeoff, interpolating between the
$\spread^{-1}$-FEF$1$ guarantee of
\Cref{thm:greedy-common} when $\eps=0$ and exact FEF$1$.

\begin{restatable}{theorem}{AugmentationUpperBound}
\label{thm:aug-ub}
Under common valuations, for every $\varepsilon\ge0$,
\textsc{Aug-GreedyFit} returns an allocation $A$
satisfying $s_i(A_i)\le(1+\varepsilon)B_i$ for every $i \in N$,
and $A$ is $\min\left\{
1,\frac{1+\varepsilon}{\Gamma}
\right\}\text{-FEF1}$ when fairness is evaluated using the original budgets. In particular, if $\varepsilon\ge\Gamma-1$, then $A$ is exact FEF1.
\end{restatable}

The proof appears in Appendix~\ref{app:proof:thm:aug-ub}.
The next theorem shows that no fixed finite extra budget overcomes the
unrestricted-density lower bound. This is stronger than ruling out only
small values of $\varepsilon$.

\begin{restatable}{theorem}{AugmentationLowerBound}
\label{thm:aug-lb}
Fix an integer $k\ge0$, a factor $\alpha\in(0,1]$, and any
fixed $\varepsilon\ge0$. Every deterministic online algorithm
maintaining \(s_i(A_i^t)\le(1+\varepsilon)B_i\) for every agent
\(i\in N\) and every prefix \(t\) fails to guarantee $\alpha$-FEF$k$ on some fixed
finite input sequence, even against an oblivious adversary.
This holds with two agents, equal unit budgets, common
valuations, uniform good sizes, a known horizon, and fairness
evaluated using the original budgets.
\end{restatable}

The proof appears in Appendix~\ref{app:proof:thm:aug-lb}.
Theorem~\ref{thm:aug-lb} shows that no finite augmentation
$\varepsilon$ yields a positive instance-independent guarantee
without bounding the density spread. Under $\Gamma\le\gamma$,
however, Theorem~\ref{thm:aug-ub} achieves exact FEF1 whenever
$\varepsilon\ge\gamma-1$.
Appendix~\ref{sec:density} further sharpens this frontier under
common valuations, common sizes, and unit budgets: as the maximum
item size tends to zero, the optimal limiting guarantee is
$\min\{1,\frac{1+\varepsilon}{1+\ln\gamma}\}$ in both the known- and unknown-horizon models. Thus, every
$\varepsilon>\ln\gamma$ suffices for exact FEF1 once goods are
sufficiently small. Appendix~\ref{app:counter} shows that this
stronger guarantee relies on the threshold policy, since
\textsc{Aug-GreedyFit} is not exact for density spread exceeding
$1$, even with arbitrarily small goods.

\section{Learning-Augmented Online Algorithms}
\label{sec:la}
%====================================================================
We next overcome the limitations of the online model using
predictions about future arrivals. Our algorithm reserves
capacity for a planned offline allocation while remaining
robust to prediction errors. Unlike preceding results, its guarantees are independent
of the density spread and instead degrade gracefully with
prediction error. 

\subsection{The Type-Count Prediction Model}

Goods with identical values but different
sizes may play different roles in feasible allocations. We therefore
predict their joint value-size profiles. Formally, fix a finite type
set $\mathcal T$ containing every predicted and realized type. A
\emph{type} is a value-size profile
$\theta=((v_{i,\theta})_{i\in N},(s_{i,\theta})_{i\in N})
\in\mathbb R_{\ge0}^n\times\mathbb R_{>0}^n$.
Under common valuations, $v_{i,\theta}=v_\theta$ for every $i\in N$.
The predictor supplies type counts
$\widehat f=(\widehat f_\theta)_{\theta\in\mathcal T}$, while
$f=(f_\theta)_{\theta\in\mathcal T}$ denotes the realized counts.
An arriving good has type $\theta$ iff its values and sizes match
$\theta$. The predicted multiset $\widehat M$ contains
$\widehat f_\theta$ copies of each type $\theta$. An offline planner
computes a feasible allocation
$\widehat A=(\widehat A_r)_{r\in\Rec}$ with quotas
$q_{r,\theta}$ satisfying
$\sum_{r\in\Rec}q_{r,\theta}=\widehat f_\theta$.

We measure prediction quality by the total value of goods whose type counts are mispredicted. The value-weighted prediction error for agent $i$ and the overall error are $\Delta_i(f,\widehat f)=\sum_{\theta\in\mathcal T}\val_{i,\theta}\,|f_\theta-\widehat f_\theta|$ and $\Delta_\infty(f,\widehat f)=\max_{i\in\N}\Delta_i(f,\widehat f)$. Under common valuations,
$\Delta(f,\widehat f)=
\sum_\theta v_\theta|f_\theta-\widehat f_\theta|$.

Given the predicted multiset, \textsc{Plan-Reserve-Fulfill}
(\Cref{alg:prf}) computes a target feasible allocation.
Online, it realizes this plan whenever possible by reserving
capacity for planned assignments, greedily allocating excess
goods whenever doing so preserves the reserved capacity, and
otherwise sending them to charity. A
realized type-$\theta$ good is \emph{excess} if it arrives
after all planned type-$\theta$ quotas have been filled.
For each recipient $r$ and type $\theta$, the algorithm
maintains the remaining planned quota
$q^{\mathrm{unf}}_{r,\theta}$, initialized to
$q_{r,\theta}$. For each agent $i$, it maintains the
\textit{reserved size}
$\mathrm{res}_i=\sum_{\theta\in\mathcal T}
q^{\mathrm{unf}}_{i,\theta}s_{i,\theta}$, the total size of planned goods assigned to $i$ whose quotas
remain unfilled, while ensuring
$s_i(A_i)+\mathrm{res}_i\le B_i$.

The online phase runs in $\mathcal{O}(n)$ time per arriving good once the offline quotas have been computed. The preprocessing time equals that of the chosen planner, e.g., pseudo-polynomial using the exact FEFx algorithm of \citet{barman2023guaranteeing} or polynomial using their FPTAS.

The key invariant is that reserved capacity never compromises
feasibility: reserved goods always fit within the budget, and
every good assigned to satisfy a planned quota is accepted
immediately. The proof appears in
Appendix~\ref{app:proof:lem:reserve}.

\begin{algorithm}[t]
\caption{\textsc{Plan-Reserve-Fulfill}}
\label{alg:prf}
\begin{algorithmic}[1]
\State Run the offline planner on $\widehat M$ to get
$\widehat A$ and quotas $q_{r,\theta}$.
\State Initialize $q^{\mathrm{unf}}_{r,\theta}\leftarrow q_{r,\theta}$
for every recipient $r$ and predicted type $\theta$,
$A_r\leftarrow\varnothing$ for all $r$, and
$\mathrm{res}_i\leftarrow
\sum_{\theta}q^{\mathrm{unf}}_{i,\theta}s_{i,\theta}$
for every agent $i$. When an unpredicted type $\theta$ first appears,
initialize $q^{\mathrm{unf}}_{r,\theta}\leftarrow0$ for every $r$.
\For{each arriving good $g$ of realized type $\theta$}
    \If{$q^{\mathrm{unf}}_{i,\theta}>0$ for some agent $i$}
        \State Pick such an agent via a fixed tie-breaking rule.
        \State Assign $g$ to $i$.
        \State $q^{\mathrm{unf}}_{i,\theta}\gets
        q^{\mathrm{unf}}_{i,\theta}-1$, $\mathrm{res}_i \leftarrow \mathrm{res}_i-s_{i,\theta}$.
    \ElsIf{$q^{\mathrm{unf}}_{\charity,\theta}>0$}
        \State Assign $g$ to charity and update $q^{\mathrm{unf}}_{\charity,\theta}\gets
        q^{\mathrm{unf}}_{\charity,\theta}-1$.
    \ElsIf{$\exists i \in N: s_i(A_i)+s_{i,\theta}+\mathrm{res}_i\le B_i$}
        \State Pick such an agent via a fixed tie-breaking rule.
        \State Assign $g$ to $i$.
    \Else { assign $g$ to charity.}
    \EndIf
\EndFor
\State \Return $(A_1,\dots,A_n,A_{\charity})$.
\end{algorithmic}
\end{algorithm}

\begin{restatable}{lemma}{ReserveFeasibility}
\label{lem:reserve}
Throughout the execution of
\textsc{Plan-Reserve-Fulfill}, the reserved-capacity invariant is
maintained. Consequently, the algorithm always returns a
budget-feasible allocation, and every good assigned to satisfy a
planned agent quota is accepted without violating feasibility.
\end{restatable}

\subsection{Consistency under Perfect Predictions}

A fundamental requirement of learning-augmented algorithms is
\emph{consistency}: under perfect predictions, the online algorithm
should reproduce the planned offline allocation. For
\textsc{Plan-Reserve-Fulfill}, this means that every predicted quota
is eventually filled, so the online allocation matches the planned
allocation up to the identities of identical goods.

\begin{restatable}{theorem}{PerfectPredictionConsistency}
\label{thm:consistency}
If $f=\widehat f$, then \textsc{Plan-Reserve-Fulfill} fills every
planned quota and returns an allocation with the same per-recipient
type counts as the planned allocation $\widehat A$. 
\end{restatable}

See Appendix~\ref{app:proof:thm:consistency} for a proof.
Theorem~\ref{thm:consistency} immediately lifts any offline fairness
guarantee determined solely by the per-recipient type counts to the
online setting. In particular, FEF, FEF$k$, $\alpha$-FEF$k$, additive
$\eta$-FEF$k$, and FEFx are all preserved; see
Appendix~\ref{app:proof:cor:fefx-lift}.

\subsection{Robustness to Prediction Errors}
\label{sec:Robustness}

We now quantify the effect of prediction errors. Planned goods
may fail to arrive, leaving quotas unfilled, while additional
realized goods create excess allocations. We measure additive
fairness loss by the \emph{$k$-violation}
$\operatorname{viol}_k(A)$, the minimum $\eta\ge0$ such
that $A$ is additive $\eta$-FEF$k$; see
Appendix~\ref{app:k-violation} for an equivalent explicit
expression. The key step is to relate this fairness loss to
the missing and excess goods caused by prediction errors.
Appendix~\ref{app:proof:thm:robust} establishes this
intermediate bound, which yields the following robustness
guarantee after relating those quantities to the prediction
error $\Delta_\infty(f,\widehat f)$; see
Appendix~\ref{app:proof:thm:robust-global} for the proof.

\begin{restatable}{theorem}{PredictionErrorRobustness}
\label{thm:robust-global}
Suppose the planned allocation $\widehat A$ is additive
$\eta_0$-FEF$k$, and let $A$ be the allocation returned by
\textsc{Plan-Reserve-Fulfill}. Then,
$\operatorname{viol}_k(A)\le
\eta_0+\Delta_\infty(f,\widehat f)$.
In the common-valuation case,
$\operatorname{viol}_k(A)\le
\eta_0+\Delta(f,\widehat f)$.
\end{restatable}

Combining \Cref{thm:consistency,thm:robust-global},
\textsc{Plan-Reserve-Fulfill} recovers the planner's
allocation under perfect predictions and otherwise degrades
gracefully with prediction error. Appendix~\ref{app:weak-advice}
shows that weaker advice consisting only of the minimum
positive density and a valid spread bound does not improve
the optimal spread-bounded frontier.
\subsection{Tightness of the Prediction-Error Bound}

The robustness guarantee of \Cref{thm:robust-global} degrades
additively with the prediction error. The following theorem shows that
this dependence is essentially tight: no deterministic
prediction-augmented online algorithm can guarantee a substantially
smaller additive dependence on the prediction error while remaining
perfectly consistent.

\begin{restatable}{theorem}{PredictionErrorLowerBound}
\label{thm:pred-lb}
Fix $k\ge1$ and $\zeta\in(0,1)$. Any
deterministic online algorithm that returns an exact
FEF$k$ allocation under correct type-count predictions satisfies $\operatorname{viol}_k(A) \ge (1-\zeta)\Delta(f,\widehat f)$ on some instance,
even for two agents with equal budgets, common
valuations, sizes common across agents, and a known horizon.
\end{restatable}

The proof appears in Appendix~\ref{app:proof:thm:pred-lb}.
The lower bound uses only exact FEF$k$ under correct predictions; it does
not assume that the algorithm follows a particular offline plan. On the
correct instance, exact FEF$k$ forces each agent to accept almost all of the
low-valued prefix. Replacing the predicted zero-valued suffix by
high-valued goods then leaves too little capacity to accept those goods,
creating a fairness violation proportional to the value-weighted prediction
error.

\subsection{Value and Size Marginals Do Not Suffice}

The following result shows that exact FEF$k$ cannot be achieved
using only value and size marginals, even when several natural
aggregate statistics are also predicted. Budget feasibility
depends on the correlation between values and sizes, making
joint type information essential.

\begin{restatable}{theorem}{MarginalsDoNotSuffice}
\label{thm:marginals}
Fix $k\ge1$. No deterministic online algorithm can guarantee
exact FEF$k$ given only the exact value and size multisets and
the horizon, even for two agents with equal budgets, common
valuations, and sizes common across agents. The result
continues to hold if the total value and maximum item value
are also provided.
\end{restatable}
\emph{Proof sketch.}
Appendix~\ref{app:proof:thm:marginals} constructs two
instances with identical marginal predictions but different
value-size pairings, forcing different online decisions.

Unlike unconstrained online fair division, where predicting
future values can suffice \cite{neoh2026online}, budget
constraints make the pairing of values and sizes essential.
Theorem~\ref{thm:marginals} thus explains why our
prediction model uses joint value-size types instead of
separate value and size marginals. Joint type counts provide
one sufficient representation of this information, though not
the only possible one.

%====================================================================

\section{Conclusion and Future Work}
\label{sec:conclusion}

We initiated the study of online fair division under budget
constraints. We proved that deterministic online algorithms
admit no meaningful fairness guarantees in full generality,
and showed how bounded density spread, resource
augmentation, and learning-augmented algorithms restore
tractability. Together, these results delineate the boundary
between impossibility and tractability.
Several directions remain for future work. It would be interesting to
extend the sharp frontier characterizations beyond the small-item
setting, particularly to arbitrary and agent-specific item sizes, and
to better understand the power of randomization, recourse, richer
arrival models, and other forms of resource augmentation. It also
remains to develop richer prediction models, and to establish tighter bounds on the predictive information
required to recover strong fairness guarantees.

\bibliographystyle{plainnat} 
\bibliography{abb,bib}

\clearpage
\appendix

\section{Additional Related Work}
\label{supp:related}

\subsection{Online Fair Division}
\label{supp:online-fair-division}

As discussed in Section~\ref{sec:related work}, our work is closely
related to the literature on online fair division, where goods arrive
sequentially and must be allocated without knowledge of future arrivals.
The area was introduced by \citet{AleksandrovEtAl2015}, motivated by
applications such as food banks, and has since grown substantially.

A large body of work studies the interplay between fairness and
efficiency in online allocation. Beyond the original food bank model, which assumes binary additive valuations~\cite{AleksandrovEtAl2015},
subsequent works consider general additive valuations and investigate fairness notions such as vanishing envy~\cite{benade2018make}, the
compatibility of fairness and approximate Pareto
efficiency~\cite{zeng2020fairness}, and dynamic information models in which only partial information about future items is
available~\cite{benade2025dynamic}. Other works instead optimize social welfare objectives such as egalitarian welfare under online
arrivals~\cite{springer2022online}, while additional variants study
dynamic reallocations~\cite{friedman2015dynamic,he2019achieving,li2018dynamic}, the use of reordering buffers \cite{amanatidis2026buffers}, restricted valuations structures \cite{amanatidis2025onlineFD-2valued,wang2026online}
and related online allocation models.

These works differ fundamentally from our setting in two respects.
First, they assume that every arriving good is allocated to some agent,
whereas budget constraints in our model may require goods to be assigned
to charity. Second, fairness is evaluated with respect to entire bundles,
whereas our model compares only budget-feasible subsets of recipients'
bundles. Consequently, the techniques and guarantees developed for
classical online fair division do not directly extend to the
budget-constrained setting considered in this paper.

\subsection{Semi-Online Fair Division}
\label{sec:Semi-Online Fair Division}

A rapidly growing line of work studies \emph{semi-online} (or
learning-augmented) fair division, where the online algorithm has access
to side information about future arrivals, typically obtained from
historical data. Much of this literature considers divisible goods and
uses predictions to improve welfare objectives such as Nash social
welfare or more general welfare functions. Existing prediction models
include monopolist-value predictions~\cite{banerjee2022online},
normalization information~\cite{barman2022universal,huang2025long},
and related forms of aggregate information about future
items~\cite{cohen2023general,an2024best}.

Prediction models have also been studied for indivisible goods.
\citet{zhou2023multi} initiated the study of learning-augmented online
allocation of indivisible goods and chores under normalized valuations,
focusing on maximin share (MMS) fairness. Building on this work,
\citet{neoh2026online} investigated predictions of aggregate values and
value frequencies for achieving fairness notions such as MMS and EF1,
while several subsequent works considered related prediction models for
online allocation of indivisible
items~\cite{spaeh2023online,balkanski2023strategyproof,cohen2024plant}.

More recently, \citet{choo2026approximate} studied PROP1 in online fair
division with predictions. Besides establishing impossibility results
against adaptive adversaries, they considered prediction models based on
maximum item values and showed that these suffice to recover meaningful
PROP1 guarantees, while stronger notions such as EF1 and MMS remain
inapproximable under the same information model. Finally,
\citet{wang2026online} investigated semi-online allocation of
indivisible goods and chores under additive and submodular valuations,
obtaining fairness and efficiency guarantees in a substantially weaker
adversarial model than the fully adaptive setting considered here.

Our learning-augmented framework differs fundamentally from these
approaches. Rather than predicting aggregate values, normalization
information, value frequencies, or maximum item values, we predict
\emph{joint value-size types}. This richer prediction model is
necessitated by budget constraints: feasibility depends not only on the
marginal distributions of values and sizes, but also on their
correlation. Our necessity result shows that these correlations cannot,
in general, be recovered from separate value and size predictions.

\subsection{Temporal Fair Division}

Recently, \citet{elkind2025temporal} introduced the model of
\emph{temporal fair division}, in which indivisible goods or chores
arrive sequentially according to a fixed sequence that is known in
advance, and the objective is to maintain approximately fair allocations
throughout the process. \citet{cookson2025temporal} study a closely
related model, considering alternative fairness notions and their
simultaneous satisfaction over allocation prefixes and in the final
allocation.
Recently, \citet{Goldberg2026} study the problem of minimizing cumulative envy in this setting.

These models differ fundamentally from ours in the information
available to the algorithm. Temporal fair division assumes complete
knowledge of all future arrivals and agents' valuations, allowing the
allocation to be planned over the entire sequence. In contrast, we study
the online setting, in which future goods are unknown unless explicitly
predicted, and allocation decisions must be made irrevocably as goods
arrive. Moreover, temporal fair division does not consider the budget
constraints and charity-based feasible fairness notions that are central
to our work.

\subsection{Online Matching}

Online matching is one of the classical problems in online algorithms,
originating with the seminal work of \citet{karp1990optimal}, who
introduced the online bipartite matching problem and the Ranking
algorithm with its optimal competitive ratio of $1-1/e$. A large body
of subsequent work has extended this model in several directions,
including weighted matchings~\cite{feldman2009online}, fully online
matching where both sides of the graph arrive over
time~\cite{wang2015two,huang2018match}, repeated matching \cite{caragiannis2023repeatedmatching,Lim2026repeated}, and numerous other variants (e.g., slot assignment \cite{elkind2022temporalslot}).

Our work shares the irrevocable decision-making paradigm of online
matching: arriving objects must be assigned without knowledge of future
arrivals. However, the optimization objective is fundamentally
different. Online matching seeks to maximize the size or weight of the
matching, whereas our goal is to produce fair allocations under
generalized assignment budget constraints. Consequently, the notions of
fairness, the role of budget feasibility, and the resulting algorithmic
techniques differ substantially from those studied in the online
matching literature.

\subsection{Additional Related Online Problems}

Other related online settings include online coalition formation \cite{cohen2024onlinefriends,cohen2024online,cohen2025online,cohen2025decentralized,cohen2026delayed,cohen2023online,cohen2025fair},  concerns forming or learning partitions of agents, and temporal voting \cite{alouf2022better,ElkindObraztsovaTeh2024TemporalFairness,teh2026price,phillips2026strengthening,elkind2025not,elkind2025verifying,elkind2024temporal,zech2024multiwinner}, which concerns repeated
collective choices and notions of representation or fairness across time.
These models are related to ours primarily through their sequential
decision-making structure.
In contrast, our setting requires the irrevocable allocation of arriving indivisible goods under per-agent budgets, with unallocated goods retained by charity and fairness evaluated through budget-feasible comparisons.

\section{Omitted Proofs for Section~\ref{sec:lb}}
\label{app:omitted:lb}

\subsection{Proof of Theorem~\ref{thm:main-lb}}
\label{app:proof:thm:main-lb}
\MainImpossibility*

\begin{proof}
Fix $k\ge 0$ and $\alpha\in(0,1]$. We first describe a phase simulation that identifies a hard
stopping phase. Since the algorithm is deterministic, the stopping
phase identified by this simulation determines a fixed finite input
sequence; we formalize this point at the end of the proof.
We construct an adversarial
instance in which every good has the same size $1/Q$, where $Q$ is
chosen sufficiently large. Since each agent has budget $1$, no agent
can receive more than $Q$ goods.

Choose an integer $Q>k$ sufficiently large that
\begin{equation}
\label{eq:Qtau}
\tau
:=
\left\lfloor
\frac{\alpha(Q-k)}{4}
\right\rfloor
\ge 1.
\end{equation}
Such a choice is possible because
$\alpha(Q-k)/4\rightarrow\infty$ as $Q\rightarrow\infty$.
The quantity $\tau$ will serve as the minimum number of goods that
each agent must receive from every phase in order to avoid an immediate
violation of $\alpha$-FEF$ k$. Note $\tau\ge 1$
is an integer and, by the definition of the floor
function,
\begin{equation}
    \label{eq:tau}
    \tau
    \le
    \frac{\alpha(Q-k)}{4}.
\end{equation}

Choose an integer $L$ satisfying
\begin{equation}
\label{eq:Ltau}
L\tau>Q.
\end{equation}
Such a choice is possible because $\tau\ge1$. We will partition the construction into $L$ phases and set the (known) horizon to $\horizon=3QL$
so that each phase consists of exactly $3Q$ goods.

It remains to specify the values of the goods released in each phase.
Let $V_1,\dots,V_L$ be positive numbers defined recursively by
setting $V_1=1$ and, for every $t\ge2$, choosing $V_t$ large
enough that
\begin{equation}
\label{eq:phaseval}
Q\sum_{r<t}V_r
<
\frac{\alpha(Q-k)}{4}\,V_t.
\end{equation}
Thus, the total value of all goods that an agent could possibly retain
from earlier phases is dominated by the value contributed by only a
constant fraction of the goods released in phase $t$.

Such a sequence always exists. For example, one may take
\[
V_t=
\frac{8Q}{\alpha(Q-k)}
\sum_{r<t}V_r,
\quad t\ge2,
\]
which satisfies~\eqref{eq:phaseval} since it  gives
$Q\sum_{r<t}V_r=\tfrac{\alpha(Q-k)}{8}V_t<\tfrac{\alpha(Q-k)}{4}V_t$. In the remainder of the proof we use only
property~\eqref{eq:phaseval}, not the explicit construction.

For each phase $t\in[L]$, choose a positive number $d_t$ satisfying
\begin{equation}
\label{eq:dummy}
Qd_t
<
\frac{\alpha(Q-k)}{4}V_t.
\end{equation}
Such a choice is possible because $V_t>0$. These values will be used only if the adversary terminates the phase construction early.

The adversary proceeds in phases. During phase $t$, it releases $3Q$ goods, each of value $V_t$ and size $1/Q$. After the phase is completed, it checks whether the algorithm satisfies the condition described below. If so, the adversary proceeds to the next phase. If not, it terminates the construction by replacing all remaining phases  with dummy goods of value $d_t$ and size $1/Q$. Since exactly $3QL-3Qt$ goods remain to be released, the total number of goods is $\horizon=3QL$, regardless of the phase at which the construction terminates.

The first observation is that, after every completed phase, at least
$Q$ goods from that phase must be assigned to charity. Indeed, phase
$t$ contains $3Q$ goods, whereas the two agents together can hold at
most $2Q$ goods in total because every good has size $1/Q$ and both
budgets are equal to $1$. Consequently, at least $3Q-2Q=Q$
goods released during phase $t$ must remain in charity.

Suppose that, after completing phase $t$, one of the two agents,
say agent $i \in \{1,2\}$, has received fewer than $\tau$ goods released during
that phase. The adversary then terminates the phase construction and
completes the remaining rounds using dummy goods of value $d_t$.
We show that the resulting allocation violates
$\alpha$-FEF$ k$.

By the previous observation, charity contains at least $Q$ goods from
phase $t$. Let $S$ be any such set. Since every good has size
$1/Q$, $\size_i(S)=Q\cdot\frac1Q=1=\bud_i$,
so $S$ is feasible for agent $i$. Moreover, for every
$X\subseteq S$ with $|X|\le k$,
\begin{equation}
\label{eq:Sval}
\val(S\setminus X)
\ge
(Q-k)V_t,
\end{equation}
because $S\setminus X$ contains at least $Q-k$ goods, each of value
$V_t$.

It remains to upper-bound the final value $v(A_i)$ accumulated by agent $i$. The
goods held by $i$ fall into three categories: goods received during
phase $t$, goods received in earlier phases, and dummy goods released
after the construction terminates.

First, by assumption, agent $i$ receives fewer than $\tau$ goods
from phase $t$. Hence, the total contribution of these goods is less
than
\[
\tau V_t
\le
\frac{\alpha(Q-k)}{4}V_t,
\]
where the inequality follows from the definition of $\tau$, which satisfies $\tau\le\alpha(Q-k)/4$ by \eqref{eq:tau}.

Second, agent $i$ holds at most $Q$ goods in total. If $n_r \le Q$
denotes the number of goods retained from phase $r<t$, then
$\sum_{r<t}n_r\le Q$, and therefore
\[
\sum_{r<t}n_rV_r
\le
Q\sum_{r<t}V_r
<
\frac{\alpha(Q-k)}4V_t,
\]
where the final inequality follows from
\eqref{eq:phaseval}.

Finally, every dummy good has value $d_t$, and agent $i$ can hold
at most $Q$ goods altogether. Thus, the total value contributed by
dummy goods is at most $Qd_t
<
\frac{\alpha(Q-k)}4V_t$ by~\eqref{eq:dummy}.

Combining the three estimates gives us
\begin{equation}
\label{eq:Aival}
\val(A_i) <
3\cdot
\frac{\alpha(Q-k)}4V_t =
\frac{3\alpha(Q-k)}4V_t <
\alpha(Q-k)V_t.
\end{equation}

Together with~\eqref{eq:Sval}, this implies that for every
$X\subseteq S$ with $|X|\le k$,
\[
\val(A_i)
<
\alpha(Q-k)V_t
\le
\alpha\,\val(S\setminus X).
\]
Hence, no set $X\subseteq S$ with $|X|\le k$ satisfies the
$\alpha$-FEF$k$ inequality for the comparison set $S$, so \textbf{the final
allocation violates $\alpha$-FEF$ k$.}

To complete the argument, suppose that the failure condition never
occurs. Then, after every phase $t\in[L]$, each agent must hold at
least $\tau$ goods released during that phase. Since goods released in
different phases are distinct, by the end of phase $L$ each agent
holds at least $L\tau$ goods. By~\eqref{eq:Ltau}, $L\tau>Q$,
contradicting the fact that each agent can hold at most $Q$ goods.
Therefore, the failure condition must occur after some phase
$t\le L$, and the argument above shows that the resulting allocation
violates $\alpha$-FEF$ k$.

Although the phase construction was described adaptively, it yields
a fixed hard input for the deterministic algorithm. Simulate the
construction above and let $t^\star$ be the first phase at which the
failure condition occurs. Before the actual execution, fix the input
sequence consisting of phases $1,\ldots,t^\star$, followed by exactly
$3Q(L-t^\star)$ dummy goods of value $d_{t^\star}$ and size $1/Q$.

A fresh execution of the algorithm on this fixed sequence has the
same announced horizon $T=3QL$ and the same observed arrival prefix
through phase $t^\star$ as the simulation. Determinism therefore
forces the algorithm to make the same assignments through that phase.
The violation established above consequently occurs on this fixed input sequence. Hence, the lower bound holds even against an oblivious adversary.

Finally, observe that the constructed instance satisfies all the restrictions stated in the theorem. There are two agents with equal budgets $\bud_1=\bud_2=1$; every good has the same size
$1/Q$; valuations are common across agents; and the horizon is fixed to $\horizon=3QL$ and announced in advance. Thus, the impossibility holds even under assumptions~\textup{(1)}--\textup{(5)}.
\end{proof}

\subsection{Proof of Corollary~\ref{cor:kappa}}
\label{app:proof:cor:kappa}
\CardinalityImpossibility*

\begin{proof}
Apply \Cref{cor:exact-fefk} with the value $Q$
chosen in the proof of \Cref{thm:main-lb}, and let
$q=Q$. Since every good has size $1/Q$ and every budget
is $1$, a set $S$ is feasible for an agent if and only if
$|S|\le Q$. For an instance $I$, recall that its \emph{maximum feasible comparison
cardinality} is $\kappa(I) := \max_{i\in\N} \max\{ |S|: S\subseteq G,\size_i(S)\le\bud_i\}$. Hence, every hard instance satisfies
$\kappa(I)=Q=q$, and the same construction establishes
the claimed impossibility.
\end{proof}

\section{Omitted Proofs for Section~\ref{sec:greedy}}
\label{app:omitted:greedy}

\subsection{Proof of Theorem~\ref{thm:greedy-common}}
\label{app:proof:thm:greedy-common}
\DensitySpreadGuarantee*

\begin{proof}
Fix an agent $i$, a recipient $r$, and a nonempty feasible
set $S\subseteq A_r$. If $\val(S)=0$, then
$\val(S\setminus\{x\})=0$ for every $x\in S$, and the
claim is immediate. Assume therefore that $\val(S)>0$.

Let $f$ be the \textit{last-arriving} good in $S$ with
$\val(f)>0$, and set $x=f$. Every positively valued
good in $S\setminus\{f\}$ arrived before $f$, while
zero-value goods do not affect the value of
$S\setminus\{f\}$ under $\val(\cdot)$. We distinguish two cases according to
whether $f$ was feasible for agent $i$ when it arrived. For the analysis below, we write $A_h^{<f}$ for the bundle
held by agent $h$ immediately before processing a good
$f$.

\paragraph{Case 1: $f$ was feasible for agent $i$ when it arrived.} That is, suppose that $\size_i(A_i^{<f})+\size_i(f)\le\bud_i$.
Then, $i\in F_f$. Since $\val(f)>0$  and
$F_f\ne\varnothing$,
\textsc{Positive-GreedyFit} assigned $f$ to an agent rather
than to charity, so $r=j$ for some agent $j\in\N$ and thus $f\in A_j$.
Since both $i$ and $j$ were feasible when $f$
arrived, and $f$ was assigned to agent $j$, the greedy min-value choice in line \ref{state:min} implies $\val(A_j^{<f})
\le
\val(A_i^{<f})$.
Every positively valued good in $S\setminus\{f\}$ lies in $A_r=A_j$ and arrived
before $f$, and hence belongs to $A^{<f}_j$. Therefore
\[
\val(S\setminus\{f\})
\le
\val(A_j^{<f})
\le
\val(A_i^{<f})
\le
\val(A_i).
\]
Thus
\[
\val(A_i)
\ge
\val(S\setminus\{f\})
\ge
\spread_i^{-1}\val(S\setminus\{f\}),
\]
which is stronger than required (the factor satisfies $1\ge\spread_i^{-1}$).

\paragraph{Case 2: $f$ was not feasible for agent $i$ when it arrived.} That is, suppose instead that $\size_i(A_i^{<f})+\size_i(f)>\bud_i$.
Since $S$ is feasible for agent $i$ and $f\in S$,
\[
\size_i(S\setminus\{f\})
\le
\bud_i-\size_i(f)
<
\size_i(A_i^{<f}).
\]

Since \textsc{Positive-GreedyFit} assigns a good to an
agent only if it has positive value, every good in
$A_i^{<f}$ has density at least
$\densdown_i$. Hence
\[
\val(A_i^{<f})
=
\sum_{g\in A_i^{<f}}
\dens_i(g)\size_i(g)
\ge
\densdown_i\,
\size_i(A_i^{<f}).
\]
Likewise,
\[
\val(S\setminus\{f\})
\le
\densup_i\,
\size_i(S\setminus\{f\}),
\]
since every positively valued good in
$S\setminus\{f\}$ has density at most
$\densup_i$, while zero-value goods contribute nothing.
Combining these inequalities yields
\[
\val(S\setminus\{f\}) \le
\densup_i\,
\size_i(S\setminus\{f\})<
\densup_i\,
\size_i(A_i^{<f}) \le
\frac{\densup_i}{\densdown_i}
\val(A_i^{<f})=
\spread_i\,
\val(A_i^{<f}) \le
\spread_i\,
\val(A_i).
\]
Therefore, rearranging gives us
\[
\val(A_i)
\ge
\spread_i^{-1}
\val(S\setminus\{f\}).
\]

The desired inequality holds in both cases. Since
$\spread_i\le\spread$ for every agent $i$, the
allocation is $\spread^{-1}$-FEF$1$. If
$\spread=1$, then every
$\spread_i=1$, and the guarantee becomes exact
FEF$1$.
\end{proof}

\subsection{Proof of Theorem~\ref{thm:scaled}}
\label{app:proof:thm:scaled}
Common valuations require all agents to agree on the
value of every good. A natural relaxation allows agents to differ only
by an individual scaling factor while preserving the common ranking of
goods. Formally, each agent $i$ has valuation
$\val_i(g)=\beta_i\val(g)$, where $\val$ is a common base
valuation and $\beta_i>0$ is an agent-specific scaling factor.

The same algorithm from Section \ref{sec:Common valuations} can be applied without modification by using the
common base valuation $\val$ in the greedy algorithm. Indeed, multiplying
all values of a single agent by the positive constant $\beta_i$ does
not change the relative ordering of that agent's goods. Moreover, the
density spread is unchanged, since $\Gamma_i = \frac{\rho_i^{\max}}{\rho_i^{\min}}$ by \Cref{def:density-spread}
and the common factor $\beta_i$ cancels. Thus, we run
\textsc{Positive-GreedyFit} exactly as in the common-valuation
setting, using the base valuation $\val$.

\begin{restatable}{theorem}{ScaledCommonGuarantee}
\label{thm:scaled}
Under scaled-common valuations, \textsc{Positive-GreedyFit} run with the base
valuation $v$, or with any positive scalar multiple of $v$, guarantees
$\Gamma^{-1}$-FEF1.
\end{restatable}
\begin{proof}
Apply Theorem~\ref{thm:greedy-common} to the common base valuation $v$.
For every agent $i$, recipient $r$, and nonempty feasible set
$S\subseteq A_r$, there is a good $x\in S$ such that $v(A_i)\ge \Gamma_i^{-1}v(S\setminus\{x\})$. Since $v_i=\beta_i v$ and $\beta_i>0$, multiplying both sides by $\beta_i$
gives $v_i(A_i)\ge \Gamma_i^{-1}v_i(S\setminus\{x\})$.
Positive scaling does not change which goods are individually feasible or
the density ratio $\Gamma_i$. Therefore, the allocation is
$\Gamma^{-1}$-FEF1.
\end{proof}

\subsection{Pseudocode of Hetero-GreedyFit}
\label{app:hgf}

\begin{algorithm}[h!]
\caption{\textsc{Hetero-GreedyFit}}
\label{alg:hgf}
\begin{algorithmic}[1]
\State Initialize $A_i\gets\varnothing$ for all $i\in\N$ and
       $A_{\charity}\gets\varnothing$.
\For{each arriving good $g$}
    \State $F_g^{+}=\{i\in\N : \size_i(A_i)+\size_i(g)\le\bud_i \ \& \ \val_i(g)>0 \}$.\label{state:eligible}
    \If{$F_g^{+}=\varnothing$} {assign $g$ to charity.}
    \Else { assign $g$ to an arbitrary $i^*\in \arg\min_{i\in F_g^{+}} \val_i(A_i)$, using a fixed tie-breaker. \label{state:min-hetero}}
    \EndIf
\EndFor
\State \Return $A=(A_1,\dots,A_n,A_{\charity})$.
\end{algorithmic}
\end{algorithm}

\subsection{Proof of Theorem~\ref{thm:hetero}}
\label{app:proof:thm:hetero}
\HeterogeneousGuarantee*

\begin{proof}
Fix an agent $i$, a recipient $r$, and a nonempty
feasible set $S\subseteq A_r$. If
$\val_i(S)=0$, then the claim is immediate. Assume
therefore that $\val_i(S)>0$.

Let $f$ denote the last-arriving good in $S$ with
$\val_i(f)>0$, and set $x=f$. Every positively
valued good in $S\setminus\{f\}$ arrived before
$f$. We distinguish two cases according to whether
$f$ was feasible for agent $i$ when it arrived.

\paragraph{Case 1: $f$ was feasible for agent $i$ when it arrived.} That is,
suppose that $\size_i(A_i^{<f})+\size_i(f)\le\bud_i$.
Then $i\in F_f^+$ since
$\val_i(f)>0$ and $f$ was feasible for agent $i$ when it arrived. Therefore, the set $F_f^+$ is nonempty, so
\textsc{Hetero-GreedyFit} assigns $f$ to an agent,
say $j$. Consequently, $r=j$ and
$f\in A_j$. Moreover,
$\val_j(f)>0$, since only positively valuing
agents belong to $F_f^+$ (i.e., \textsc{Hetero-GreedyFit} assigns $f$
to $j$ only if $j\in F_f^+$). 

Since both $i$ and $j$ belong to $F_f^+$,
the min-value choice in line \ref{state:min-hetero} gives
\begin{equation}
    \label{eq:hetgreedy}
    \val_j(A_j^{<f})
    \le
    \val_i(A_i^{<f}).
\end{equation}

Now, consider any good
$h\in S\setminus\{f\}$ with
$\val_i(h)>0$. Then $h\in A_r=A_j$ and
$h$ arrived before $f$, so $h\in A^{<f}_j$. Moreover, $h$ was assigned to $j$ only
if $j\in F_h^+$, so $\val_j(h)>0$. Therefore, the definition of $\chi_i$ implies $\val_i(h) \le \chi_i\val_j(h)$.
Summing over the (positively $i$-valued) goods of
$S\setminus\{f\}$, and using the fact that these goods are distinct members of $A^{<f}_j$
with nonnegative $\val_j$,
\[
\val_i(S\setminus\{f\}) =
\sum_{\substack{
h\in S\setminus\{f\}\\
\val_i(h)>0}}
\val_i(h) \le
\chi_i
\sum_{\substack{
h\in S\setminus\{f\}\\
\val_i(h)>0}}
\val_j(h) \le
\chi_i
\val_j(A_j^{<f}) \stackrel{\eqref{eq:hetgreedy}}{\le}
\chi_i
\val_i(A_i^{<f}) \le
\chi_i
\val_i(A_i).
\]
Hence
\[
\val_i(A_i)
\ge
\chi_i^{-1}
\val_i(S\setminus\{f\})
\ge
\Lambda_i^{-1}
\val_i(S\setminus\{f\}).
\]

\paragraph{Case 2: $f$ was not feasible for agent $i$ when it arrived.} That is,
suppose instead that $\size_i(A_i^{<f})+\size_i(f)>\bud_i$.
Exactly as in the proof of
Theorem~\ref{thm:greedy-common},
\[
\size_i(S\setminus\{f\})
<
\size_i(A_i^{<f}).
\]
\textsc{Hetero-GreedyFit} assigns a good to $i$ only when it has positive value for
agent $i$, so every good
in $A^{<f}_i$ has positive value for
agent $i$ and density at least $\densdown_i$.
Hence,
\[
\val_i(A_i^{<f})
\ge
\densdown_i
\size_i(A_i^{<f}),
\]
whereas
\[
\val_i(S\setminus\{f\})
\le
\densup_i
\size_i(S\setminus\{f\}).
\]
Therefore
\[
\val_i(S\setminus\{f\})\ <\ \densup_i\,\size_i(A^{<f}_i)
\ \le\ \spread_i\,\val_i(A^{<f}_i)\ \le\ \spread_i\,\val_i(A_i),
\]
which implies
\[
\val_i(A_i)
\ge
\spread_i^{-1}
\val_i(S\setminus\{f\})
\ge
\Lambda_i^{-1}
\val_i(S\setminus\{f\}).
\]

The desired inequality holds in both cases. Since
$\Lambda_i\le\Lambda$ for every agent, we have $\Lambda_i^{-1} \ge \Lambda^{-1}$,
which implies that
\[
\val_i(A_i)\ge \Lambda^{-1}\val_i(S\setminus\{f\}).
\]
Taking the worst case over $i$, $r$, and $S$ shows that the
allocation is $\Lambda^{-1}$-FEF$1$.
\end{proof}

\section{Omitted Proofs for Section~\ref{sec:aug}}
\label{app:omitted:aug}

\subsection{Pseudocode of \textsc{Aug-GreedyFit}$(\eps)$}
\label{app:aug-greedyfit}

\begin{algorithm}[h!]
\caption{\textsc{Aug-GreedyFit}$(\eps)$}
\label{alg:aug-gf}
\begin{algorithmic}[1]
\State Initialize $A_i\gets\varnothing$ for all $i\in\N$ and
       $A_{\charity}\gets\varnothing$.
\For{each arriving good $g$}
    \If{$\val(g)=0$} {assign $g$ to charity.}
    \Else { let
        $ F_g := \{ i\in N: s_i(g)\le B_i \ \text{and}\ s_i(A_i)+s_i(g)\le(1+\varepsilon)B_i\}$.}
        \If{$F_g=\varnothing$} {assign $g$ to charity.}
        \Else { Assign $g$ to an arbitrary
            $i^*\in
            \arg\min_{h\in F_g}\val(A_h)$,
            using a fixed tie-breaking order.}
        \EndIf
    \EndIf
\EndFor
\State \Return
$
A=(A_1,\dots,A_n,A_{\charity})
$.
\end{algorithmic}
\end{algorithm}

\subsection{Proof of Theorem~\ref{thm:aug-ub}}
\label{app:proof:thm:aug-ub}
\AugmentationUpperBound*

\begin{proof} 
The size bound $s_i(A_i)\le(1+\varepsilon)B_i$ follows directly from the feasibility test in the algorithm. It therefore remains to prove the
fairness guarantee.

Fix an agent $i$, a recipient $r$, and a nonempty
set $S\subseteq A_r$ feasible for agent $i$ under the original budget,
i.e., $\size_i(S)\le\bud_i$. 
If $v(S)=0$, choose any $x\in S$. Then
\[
v(A_i)\ge 0=v(S\setminus\{x\}),
\]
so the required inequality holds. Henceforth assume $v(S)>0$.
Let $f$ denote the
last-arriving positively valued good in $S$, and set
$x=f$. Every other positively valued good in
$S\setminus\{f\}$ arrived before $f$. As in the proof of
\Cref{thm:greedy-common}, we distinguish two cases
according to whether $f$ was feasible for agent $i$
under the algorithmic budget. As before, let
$A_h^{<f}$ be the bundle held by agent $h$ immediately
before processing a good $f$.

\paragraph{Case 1: $f$ was feasible for agent $i$ under the augmented budget at its arrival.}
That is, suppose that $\size_i(A_i^{<f})+\size_i(f)\le(1+\eps)\bud_i$.
Since $f\in S$ and $S$ is feasible for agent $i$ under the
original budget, we also have $s_i(f)\le B_i$. Hence both
eligibility conditions in Algorithm~\ref{alg:aug-gf} hold for
agent $i$, so $i\in F_f$. The min-value rule therefore assigned
$f$ to some agent $j$ with $\val(A_j^{<f})\le\val(A_i^{<f})$.
Every positively valued good of $S\setminus\{f\}$ lies in $A_r=A_j$
and arrived before $f$, and therefore belongs to $A_j^{<f}$. Hence,
\[
\val(S\setminus\{f\})\le\val(A_j^{<f})\le\val(A_i^{<f})\le\val(A_i),
\]
which is the claim with factor $1$.
Note that this is stronger than the desired guarantee.

\paragraph{Case 2: $f$ was not feasible for agent $i$ under the augmented budget at its arrival.}
That is, suppose instead that $\size_i(A_i^{<f})+\size_i(f)>(1+\eps)\bud_i$.
Then
\[
\size_i(A_i^{<f})
>
(1+\eps)\bud_i-\size_i(f).
\]
Since $S$ is feasible under the original budget,
\[
    s_i(S\setminus\{f\})\le B_i-s_i(f).
\]
If $v(S\setminus\{f\})=0$, the required inequality is
immediate. We may therefore assume that
$v(S\setminus\{f\})>0$. In particular, since every good has
positive size, $s_i(f)<B_i$.

Let $c:=s_i(f)$. Since $f\in S$ and $S$ is feasible for $i$, we have
$s_i(f)\le B_i$. Therefore, if $f$ was not eligible for $i$,
the failed eligibility condition must be the total online budget
condition:
\[
    s_i(A_i^{<f})>(1+\varepsilon)B_i-c.
\]
Every good in $A_i^{<f}$ has positive value, so
\[
    v(A_i)
    \ge v(A_i^{<f})
    \ge \rho_i^{\min}s_i(A_i^{<f})
    >
    \rho_i^{\min}((1+\varepsilon)B_i-c).
\]
On the other hand,
\[
    v(S\setminus\{f\})
    \le
    \rho_i^{\max}s_i(S\setminus\{f\})
    \le
    \rho_i^{\max}(B_i-c).
\]
Therefore,
\[
    v(A_i) >
    \frac{\rho_i^{\min}}{\rho_i^{\max}}
    \frac{(1+\varepsilon)B_i-c}{B_i-c}
    v(S\setminus\{f\}) =
    \Gamma_i^{-1}
    \left(1+\frac{\varepsilon B_i}{B_i-c}\right)
    v(S\setminus\{f\})\ge
    \Gamma_i^{-1}(1+\varepsilon)
    v(S\setminus\{f\}).
\]

In both cases,
\[
\val(A_i)
\ge
\min\{1,\spread_i^{-1}(1+\eps)\}
\val(S\setminus\{f\}).
\]
Since
$\spread_i\le\spread$,
\[
\min\{1,\spread_i^{-1}(1+\eps)\}
\ge
\min\left\{1,\frac{1+\eps}{\spread}\right\},
\]
and taking the worst case over
$i$, $r$, and $S$
proves the stated
$\min\{1,(1+\eps)/\spread\}$-FEF$1$
guarantee.

Finally, if
$\eps\ge\spread-1$, then
\[
\spread_i^{-1}(1+\eps)
\ge
\spread^{-1}(1+\eps)
\ge
1
\]
for every agent $i$. Hence, both cases yield factor
$1$, and the allocation is exactly
FEF$1$.
\end{proof}

\subsection{Proof of Theorem~\ref{thm:aug-lb}}
\label{app:proof:thm:aug-lb}
\AugmentationLowerBound*

\begin{proof}
Fix $k$, $\alpha$, and $\varepsilon$. As in the proof of \Cref{thm:main-lb}, we first use a phase
simulation to identify a hard stopping phase and then fix the
resulting sequence before its actual execution.
Choose an integer
$Q>k$ sufficiently large that
\[
\tau
:=
\left\lfloor
\frac{\alpha(Q-k)}{4}
\right\rfloor
\ge 1.
\]
Every good will have size $1/Q$. Under online budget
$1+\varepsilon$, each agent can hold at most $M:=\left\lfloor(1+\varepsilon)Q\right\rfloor$ goods.

Choose an integer $L$ satisfying $L\tau>M$.
Each phase will contain $P:=2M+Q$ goods, and the known horizon is $T:=LP$.

Let $V_1:=1$. For every $t\ge 2$, choose $V_t>0$ so that
\[
M\sum_{r<t}V_r
<
\frac{\alpha(Q-k)}{4}V_t.
\tag{1}
\label{eq:extra-budget-phase-growth}
\]
For each $t\in[L]$, also choose $d_t>0$ so that
\[
Md_t
<
\frac{\alpha(Q-k)}{4}V_t.
\tag{2}
\label{eq:extra-budget-dummy}
\]

In phase $t$, release $P$ goods, each having common value $V_t$
and common size $1/Q$. After the phase, check whether each agent
has received at least $\tau$ goods from that phase. If some agent
has received fewer than $\tau$ such goods, fill every remaining
round with goods of common value $d_t$ and common size $1/Q$.

First observe that charity receives at least $Q$ goods from every
completed phase. Indeed, the two agents together can hold at most
$2M$ goods over the entire execution, whereas the phase contains
$2M+Q$ goods.

Suppose that after phase $t$, agent $i$ has received fewer than
$\tau$ goods from that phase. Let $S$ be any $Q$ phase-$t$
goods in charity. Then $S$ is feasible under the original budget,
because $s_i(S)=Q\cdot\frac1Q=1$.
For every $X\subseteq S$ with $|X|\le k$,
\[
v(S\setminus X)\ge (Q-k)V_t.
\tag{3}
\label{eq:extra-budget-charity-value}
\]

We now bound the final value of agent $i$. Its phase-$t$ goods
contribute less than
\[
\tau V_t
\le
\frac{\alpha(Q-k)}{4}V_t.
\]
The agent holds at most $M$ goods in total, so its goods from
earlier phases contribute at most
\[
M\sum_{r<t}V_r
<
\frac{\alpha(Q-k)}{4}V_t
\]
by \eqref{eq:extra-budget-phase-growth}. Its later dummy goods
contribute at most
\[
Md_t
<
\frac{\alpha(Q-k)}{4}V_t
\]
by \eqref{eq:extra-budget-dummy}. Therefore,
\[
v(A_i)
<
\frac{3\alpha(Q-k)}{4}V_t
<
\alpha(Q-k)V_t.
\]
Together with \eqref{eq:extra-budget-charity-value}, this gives
\[
v(A_i)<\alpha v(S\setminus X)
\]
for every $X\subseteq S$ with $|X|\le k$. Hence the final
allocation is not $\alpha$-FEF$k$.

It remains to show that such a phase must occur. If it never occurs,
then each agent receives at least $\tau$ goods in every phase.
After $L$ phases, each agent therefore holds at least
$L\tau>M$ goods, contradicting the definition of $M$.

To see that the hard input may be fixed in advance, simulate the
phase construction and let $t^\star$ be the first phase at which an
agent receives fewer than $\tau$ goods from that phase. Fix the
sequence consisting of phases $1,\ldots,t^\star$, followed by
exactly $P(L-t^\star)$ dummy goods of value $d_{t^\star}$ and size
$1/Q$. A fresh execution has the same announced horizon and the same
arrival prefix through phase $t^\star$, so determinism forces the
same assignments and hence the same violation. Thus, the lower bound
holds against an oblivious adversary.

All goods have common value across the two agents and uniform size
$1/Q$, the original budgets are equal to $1$, and the announced
horizon is $T=LP$. This proves the theorem.
\end{proof}

\section{Small-Good Limits under Common Valuations and Sizes
Common Across Agents}\label{sec:density}\label{sec:sharp-frontier}
%====================================================================

Sections~\ref{sec:greedy} and~\ref{sec:aug} establish broad greedy guarantees that apply to arbitrary item sizes, but they leave a quantitative gap
between the corresponding upper and lower bounds. In this appendix, we
show that this gap disappears under a more structured model consisting
of common valuations, sizes common across agents, and globally small
goods. In this setting, the limiting deterministic approximation frontier can be
characterized exactly.
Our analysis pairs a threshold-based algorithm with a matching
multi-agent lower bound, yielding a sharp limiting characterization of
the optimal deterministic guarantee. The upper bound controls
charity comparisons through the threshold policy, while balanced
acceptance ensures exact FEF$1$ for all agent--agent comparisons.

The frontier results in this appendix are established for a structured
class of instances that combines the assumptions introduced
individually in earlier sections.

\begin{definition}[Structured class with bounded density spread]
\label{def:structured-class}
For $\gamma\ge 1$, let $\mathcal I^{\mathrm{com}}_\gamma$ be the class
of instances satisfying:
\begin{enumerate}
    \item valuations are common, so $v_i=v$ for every agent $i$;
    \item sizes are common across agents, so $s_i=s$ for every agent $i$;
    \item every original budget equals $1$; and
    \item the realized density spread satisfies $\Gamma\le\gamma$ according
    to Definition~\ref{def:density-spread}.
\end{enumerate}
The algorithm is told the valid bound $\gamma$. The horizon model is stated
separately in each theorem.
\end{definition}
Throughout this appendix, $\Gamma$ denotes the realized density
spread of the instance, whereas $\gamma$ denotes the valid upper
bound supplied to the algorithm. The case $\gamma=1$ is already
covered by \Cref{thm:greedy-common}, which gives exact FEF$1$ for
arbitrary item sizes. The frontier theorems below therefore focus on
$\gamma>1$.

All original budgets are normalized to $1$, and sizes are common
across agents. Hence, the globally $\sigma_{\max}$-small condition
reduces to
\[
s(g)\le \sigma_{\max}
\quad
\text{for every }g\in G.
\]
The upper bounds hold after every prefix and therefore apply under
both known and unknown horizons. For each deterministic algorithm,
the lower bounds provide a fixed finite hard sequence under either
horizon model. All algorithms are required to respect the relevant
online budget at every prefix, and all results hold for arbitrary
$n\ge1$. The lower-bound construction realizes density spread
exactly $\gamma$, so the frontier is not an artefact of supplying a
loose spread bound.

\paragraph{Density normalization.}
If there is no positively valued good that is individually feasible
under the original unit budget, then every feasible comparison set
has value zero and all claims are immediate. Otherwise, define
\[
\underline\rho
:=
\min\{\rho(g):v(g)>0,\ s(g)\le1\}.
\]
For the analysis, replace the common value function $v$ by
$v/\underline\rho$. Since valuations and sizes are common across
agents, all agents share the same density function. This single
global scaling sends every positive density of an individually
feasible good to the interval $[1,\Gamma]\subseteq[1,\gamma]$.
It scales both sides of the threshold test
$v(A_i)<\alpha\rho(g)$ by the same positive constant and preserves
the ordering of current bundle values. Consequently, it preserves
the entire execution of \textsc{Threshold}($\alpha$), including its
minimum-current-value tie-breaking rule. The normalization is only
an analytical device; the algorithm requires no knowledge of
$\underline\rho$.

We now introduce the threshold policy Threshold$(\alpha)$
(Algorithm~\ref{alg:thr}), used to establish the upper-bound frontier.
The frontier upper bound is achieved by a threshold-based allocation
algorithm. An arriving good is assigned only to agents whose current
bundle value remains below a density-dependent threshold, whose
remaining online budget is sufficient to accommodate the good, and for
whom the good is individually feasible under the original unit budget.
Among all such eligible agents, the algorithm assigns the good to one
whose current bundle has minimum value.

\begin{algorithm}[t]
\caption{\textsc{Threshold}($\alpha$) with online budget $B'$}
\label{alg:thr}
\begin{algorithmic}[1]
\State Initialize $A_i\leftarrow\varnothing$ for all $i\in N$
and $A_{\charity}\leftarrow\varnothing$.
\For{each arriving good $g$}
    \State Let
        $P_g
        :=
        \{
        i\in N:
        v(g)>0,\ s(g)\le1,\ 
        s(A_i)+s(g)\le B',
        v(A_i)<\alpha\rho(g)
        \}$.
    \If{$P_g=\varnothing$}
        \State Assign $g$ to charity.
    \Else
        \State Assign $g$ to an agent $i^*\in\arg\min_{i\in P_g}v(A_i)$         using a fixed tie-breaking order.
    \EndIf
\EndFor
\end{algorithmic}
\end{algorithm}

The key property of Threshold$(\alpha)$ is that an eligible
agent never exhausts her remaining budget, provided every
item is sufficiently small. The following lemma quantifies
this residual-capacity invariant, which underlies all of the
upper-bound results in this appendix.

\begin{lemma}[Capacity and no filling]
\label{lem:capacity-no-fill}
Assume that every positively valued good that is individually
feasible under the original unit budget has density in
$[1,\gamma]$, and that every good assigned to agent $i$ has size at most
$\sigma_{\max}$. Run \textsc{Threshold}($\alpha$) with
$\alpha\in(0,1]$ and online budget $B'\ge1$. Then $s(A_i)
\le
\alpha(1+\ln\gamma)+(1+\gamma)\sigma_{\max}$.
Consequently, if
\begin{equation}
\alpha(1+\ln\gamma)+(2+\gamma)\sigma_{\max}
\le B',
\label{eq:threshold-room}
\end{equation}
then, for every prefix $t$, $s(A_i^t)+\sigma_{\max}\le B'$.
Hence every arriving good of size at most $\sigma_{\max}$
fits in agent $i$'s remaining online budget.
\end{lemma}

\begin{proof}
If no good is assigned to agent $i$, the claim is immediate.
Otherwise, let
$g_1,\dots,g_m$
denote the goods assigned to $i$, listed in their order of
acceptance. For each $t\in[m]$, denote
\[
s_t:=s_i(g_t),\quad
\rho_t:=\rho_i(g_t),\quad
u_t:=\rho_t s_t=v(g_t),
\]
so that $s_t\le\sigma_{\max}$ by assumption. Define $w_t:=\sum_{r=1}^t u_r$ and $w_0:=0$.
For each $t\in[m]$, define $\Delta_t:=w_t-w_{t-1}=u_t$.
When $g_t$ is assigned to $i$, the threshold condition implies $w_{t-1}<\alpha\rho_t$.
Since also $\rho_t\ge1$, $\rho_t
    \ge
    \max\left\{1,\frac{w_{t-1}}{\alpha}\right\}$.
Define $f(w):=
    \frac{1}{
        \max\{1,w/\alpha\}
    }$.
Then, the function $f$ is nonincreasing, and therefore $s_t
    =
    \frac{\Delta_t}{\rho_t}
    \le
    \Delta_t f(w_{t-1})$.
Summing over all accepted goods gives
\begin{equation}
\label{eq:left-riemann-capacity}
    s_i(A_i)
    \le
    \sum_{t=1}^m
    \Delta_t f(w_{t-1}).
\end{equation}
which is the left Riemann sum for $f$.
Indeed, we now compare this left Riemann sum with the corresponding
integral. For each $t$,
\[
    \Delta_t f(w_{t-1})
    =
    \int_{w_{t-1}}^{w_t}f(w)\,dw
    +
    \int_{w_{t-1}}^{w_t}
    (f(w_{t-1})-f(w))\,dw.
\]
Since $f$ is nonincreasing,
\[
    \int_{w_{t-1}}^{w_t}
    (f(w_{t-1})-f(w))\,dw \le
    \Delta_t
    (f(w_{t-1})-f(w_t))\le
    u_{\max}
    (f(w_{t-1})-f(w_t)),
\]
where $u_{\max}:=\max_{1\le t\le m}\Delta_t$.
Summing over $t$ telescopes:
\[
    \sum_{t=1}^m
    \Delta_t f(w_{t-1})
    \le
    \int_0^{w_m}f(w)\,dw
    +
    u_{\max}(f(0)-f(w_m)) \le
    \int_0^{w_m}f(w)\,dw+u_{\max}.
\]
Since $\rho_t\le\gamma$ and $s_t\le\sigma_{\max}$, $u_{\max}\le\gamma\sigma_{\max}$.
It remains to bound the integral by first bounding its upper integration limit
$w_m$.  The last accepted good satisfies $w_{m-1}<\alpha\rho_m\le\alpha\gamma$, and hence
\[
    w_m
    =
    w_{m-1}+\Delta_m
    <
    \alpha\gamma+\gamma\sigma_{\max}
    =
    \gamma(\alpha+\sigma_{\max}).
\]
Two cases remain. If $w_m\le\alpha$, then
\[
    \int_0^{w_m}f(w)\,dw
    =
    w_m
    \le
    \alpha
    \le
    \alpha(1+\ln\gamma)+\sigma_{\max}.
\]
Otherwise, if $w_m>\alpha$, then
\[
\begin{aligned}
    \int_0^{w_m}f(w)\,dw
    &=
    \alpha+
    \int_\alpha^{w_m}\frac{\alpha}{w}\,dw\\
    &=
    \alpha+\alpha\ln\frac{w_m}{\alpha}\\
    &<
    \alpha+
    \alpha\ln\left(
        \gamma\left(1+\frac{\sigma_{\max}}{\alpha}\right)
    \right)\\
    &=
    \alpha(1+\ln\gamma)
    +
    \alpha\ln\left(
        1+\frac{\sigma_{\max}}{\alpha}
    \right)\\
    &\le
    \alpha(1+\ln\gamma)+\sigma_{\max},
\end{aligned}
\]
where the last inequality uses $\ln(1+x)\le x$.
Combining this with
$u_{\max}\le\gamma\sigma_{\max}$ in
\eqref{eq:left-riemann-capacity} gives us
\[
    s_i(A_i)
    \le
    \alpha(1+\ln\gamma)
    +(1+\gamma)\sigma_{\max}.
\]
If condition~\eqref{eq:threshold-room} holds, then $s_i(A_i)\le B'-\sigma_{\max}$.

Since the occupied size is monotone over time, the same bound
holds at every prefix. Consequently, $s_i(A_i^t)+\sigma_{\max}\le B'$ 
for every round $t$, proving that every future arrival of
size at most $\sigma_{\max}$ remains feasible for agent
$i$.
\end{proof}

The lemma uses only the sizes and densities of goods assigned to an
agent. Since Threshold$(\alpha)$ rejects individually infeasible goods,
Definition~\ref{def:density-spread} supplies the required density bounds
for every accepted good.

The previous lemma ensures that every sufficiently small arriving good
fits for every agent. The next lemma then handles comparisons
between agents: assigning each accepted good to an eligible agent
with minimum current value maintains exact FEF1 between every
ordered pair of agents.

\begin{lemma}[Minimum-value assignment gives exact
agent-agent FEF1]
\label{lem:balanced-agent-envy}
Assume common additive valuations and sizes common across
agents. Run Threshold$(\alpha)$ with a common algorithmic
budget $B'$. Suppose that, whenever a good $g$ is assigned
to an agent at round $t$, it is feasible for every agent
immediately before the assignment; that is,
\begin{equation}
\label{eq:universal-arrival-feasibility}
    s(A_i^{t-1})+s(g)\le B'
    \quad
    \text{for every }i\in N.
\end{equation}
Then, after every prefix, every ordered pair of agents $i,j$
and every nonempty set $S\subseteq A_j$
that is feasible for agent $i$, satisfy the following property: there exists $x\in S$ such that $v(A_i)\ge v(S\setminus\{x\})$.
Consequently, every agent--agent comparison satisfies exact FEF$1$, for every
number of agents and independently of $\alpha$.
\end{lemma}
\begin{proof}
Fix a prefix, an ordered pair of agents $i,j$, and suppose
$A_j\neq\varnothing$; otherwise, there is no nonempty comparison set
$S\subseteq A_j$, and the claim is vacuous.

Let $g^*$ denote the last good assigned to agent $j$, and let $A_j^{<}:=A_j\setminus\{g^*\}$.
For each agent $h$, let $w_h^{<}$ denote the value of
agent $h$'s bundle immediately before $g^*$ was assigned.

Since $g^*$ was assigned to agent $j$, the threshold rule gives $w_j^{<}<\alpha\rho(g^*)$.
We first show that $w_i^{<}\ge w_j^{<}$.
If agent $i$ was eligible to receive $g^*$, then the
smallest-current-value tie-breaking rule selected agent $j$, implying $w_j^{<}\le w_i^{<}$.

Otherwise, agent $i$ was not eligible. By
\eqref{eq:universal-arrival-feasibility}, this cannot be due to
insufficient remaining online budget. The common original-feasibility
test $s(g^*)\le1$ also holds because $g^*$ was assigned to agent $j$,
and $v(g^*)>0$. Hence the only remaining eligibility condition that can
fail is the threshold condition. Therefore,
\[
w_i^{<}\ge\alpha\rho(g^*)>w_j^{<}.
\]

Thus, $w_i^{<}\ge w_j^{<}$ in either case. Since bundle values are
nondecreasing over time, we have the following at the considered prefix
\[
v(A_i)
\ge
w_i^{<}
\ge
w_j^{<}
=
v(A_j^{<}).
\]

Now, let $S\subseteq A_j$ be any nonempty set feasible for agent
$i$.
If $g^*\in S$, choose $x=g^*$. Then $S\setminus\{x\}\subseteq A_j^{<}$,
and therefore
\[
v(S\setminus\{x\})
\le
v(A_j^{<})
\le
v(A_i).
\]

Otherwise, $g^*\notin S$, so
$S\subseteq A_j^{<}$. Hence $v(S)
\le v(A_j^{<})
\le v(A_i)$.
Choosing any $x\in S$, non-negativity of the valuation implies
\[
v(S\setminus\{x\})
\le
v(S)
\le
v(A_i).
\]

Thus, in both cases, there exists a good $x\in S$ satisfying $v(A_i)\ge v(S\setminus\{x\})$,
proving exact agent--agent FEF$1$.
\end{proof}

The hypothesis of Lemma~\ref{lem:balanced-agent-envy} is verified
through Lemma~\ref{lem:capacity-no-fill}. In Theorems~\ref{thm:limiting-density-frontier}
and~\ref{thm:limiting-augmentation-frontier}, 
every good assigned to an agent satisfies the corresponding small-item condition,
and Lemma~\ref{lem:capacity-no-fill} guarantees that $s(A_i^{t-1})+\sigma_{\max}\le B'$ for every agent $i$ at every prefix. Consequently, every arriving
good that is eligible for assignment is feasible for every agent
immediately before it is assigned, establishing condition
\eqref{eq:universal-arrival-feasibility}.

\subsection{Matching Lower Bound}

The previous lemmas establish the ingredients needed for the frontier upper bound. We now show that this guarantee is optimal in both the known-horizon and unknown-horizon models. The hard sequence below is fixed and finite, and the proof specifies how it is used in either model.

\begin{theorem}
\label{thm:small-good-lower}
Fix $n\ge 1$, $\gamma>1$, $\alpha\in(0,1]$, and
$\varepsilon\ge 0$ such that $\alpha(1+\ln\gamma)>1+\varepsilon$.
For every $\sigma_{\max}>0$ and every deterministic online
algorithm that respects online budget $1+\varepsilon$, there is a
fixed finite instance such that:

\begin{enumerate}
    \item there are $n$ agents with original unit budgets;
    \item valuations are common;
    \item every good has the same size
    $\delta\le\sigma_{\max}$;
    \item the realized positive-density spread is exactly $\gamma$;
    \item the final allocation is not $\alpha$-FEF1.
\end{enumerate}

The statement holds both when the horizon is announced before the
first arrival and when it is not announced.
\end{theorem}

\begin{proof}
For $M\ge 1$, define $S_M:=M\left(1-\gamma^{-1/M}\right)$.
Since $S_M\to\ln\gamma$, first choose $M$ and then choose an
integer $Q$ sufficiently large that, with $\delta:=\frac1Q$,
all three conditions
\begin{equation*}
\delta<\min\left\{\frac14,\sigma_{\max}\right\},
\label{eq:small-delta}
\end{equation*}
\begin{equation}
\gamma\delta<\alpha(1-2\delta),
\label{eq:small-positive-target}
\end{equation}
and
\begin{equation}
\alpha(1-2\delta)(1+S_M)-(\gamma-1)\delta
>
1+\varepsilon
\label{eq:small-capacity-contradiction}
\end{equation}
hold. Such a choice is possible because
$\alpha(1+\ln\gamma)>1+\varepsilon$.

Release one initial good $q$ having common size $\delta$,
common density $\gamma$, and common value $v(q)=\gamma\delta$.

For $j=0,\ldots,M$, define $h_j:=\gamma^{j/M}$.
A level-$j$ good has common size $\delta$, common density
$h_j$, and common value $h_j\delta$.

Set $P
:=
\left\lceil
(1+n(1+\varepsilon))Q
\right\rceil$
and $T:=1+(M+1)P$.
In the known-horizon model, announce $T$ before the first arrival.
In the unknown-horizon model, reveal no horizon.

Simulate the algorithm on the master sequence consisting of the
initial good $q$, followed by $P$ level-$0$ goods, then $P$
level-$1$ goods, and so on through $P$ level-$M$ goods.

After each level $j$, charity contains at least $Q$ level-$j$
goods. Indeed, the agents together can hold at most
$n(1+\varepsilon)Q$ goods, because every good has size $1/Q$.
Hence the number of level-$j$ goods in charity is at least
\[
P-n(1+\varepsilon)Q\ge Q.
\]
Let $S_j$ be any $Q$ such goods. Then $s(S_j)=1$, and for
every $X\subseteq S_j$ with $|X|\le 1$,
\begin{equation}
v(S_j\setminus X)
\ge
(Q-1)h_j\delta
=
h_j(1-\delta)
\ge
h_j(1-2\delta).
\label{eq:small-charity-set}
\end{equation}

We claim that there are a level $j$ and an agent $i$ such that
\begin{equation}
v(A_i^j)
<
\alpha h_j(1-2\delta),
\label{eq:small-bad-level}
\end{equation}
where $A_i^j$ denotes agent $i$'s bundle immediately after
level $j$.

Suppose otherwise. Fix an agent $i$. Let
\[
y_i
:=
\begin{cases}
\delta, & \text{if $q$ is assigned to $i$},\\
0,      & \text{otherwise},
\end{cases}
\]
and let $x_{i,\ell}$ be the total size of level-$\ell$ goods
assigned to $i$. The assumed failure of
\eqref{eq:small-bad-level} gives, for every $j=0,\ldots,M$,
\[
\gamma y_i+\sum_{\ell\le j}h_\ell x_{i,\ell}
\ge
\alpha h_j(1-2\delta).
\]
Equivalently,
\[
\sum_{\ell\le j}h_\ell x_{i,\ell}
\ge
c_{i,j},
\quad
c_{i,j}
:=
\alpha h_j(1-2\delta)-\gamma y_i.
\label{eq:small-lp-targets}
\]
Since $y_i\le\delta$, condition
\eqref{eq:small-positive-target} gives $c_{i,0}>0$.
The sequence $c_{i,0},\ldots,c_{i,M}$ is nondecreasing.
Corollary~\ref{cor:ascending-lp} therefore gives
\[
\sum_{\ell=0}^M x_{i,\ell}
\ge
\alpha(1-2\delta)(1+S_M)-\gamma y_i.
\]
Including the initial good, the total size assigned to $i$ is at
least
\begin{equation*}
y_i+\sum_{\ell=0}^M x_{i,\ell} \ge
\alpha(1-2\delta)(1+S_M)-(\gamma-1)y_i \ge
\alpha(1-2\delta)(1+S_M)-(\gamma-1)\delta
>
1+\varepsilon,
\end{equation*}
where the final inequality is
\eqref{eq:small-capacity-contradiction}. This contradicts the online
budget. Hence \eqref{eq:small-bad-level} holds.

Let $j^*$ be the first level for which
\eqref{eq:small-bad-level} holds for some agent $i$. Define the
actual input sequence as follows:

\begin{enumerate}
    \item release the initial good $q$;
    \item release all level goods through the end of level
    $j^*$;
    \item release exactly $(M-j^*)P$ zero-value goods, each
    having common size $\delta$.
\end{enumerate}

This sequence has exactly
\[
1+(j^*+1)P+(M-j^*)P
=
1+(M+1)P
=
T
\]
arrivals. In the known-horizon model, the master simulation and the
actual execution have the same announced horizon and the same
arrival prefix through level $j^*$. In the unknown-horizon
model, they also have the same observed prefix. Since the algorithm
is deterministic, it makes the same assignments through that level.

All later goods have value zero, so the failing agent's value does
not increase. By \eqref{eq:small-charity-set}, for every
$X\subseteq S_{j^*}$ with $|X|\le 1$,
\[
v(A_i)
<
\alpha h_{j^*}(1-2\delta)
\le
\alpha v(S_{j^*}\setminus X).
\]
Thus the final allocation is not $\alpha$-FEF1.

Every good has size $\delta\le\sigma_{\max}$. The initial good has
density $\gamma$, the level-$0$ goods have density $1$, and all
other positive densities lie in $[1,\gamma]$. Hence the realized
positive-density spread is exactly $\gamma$.
\end{proof}

The calibration good is introduced solely to ensure that every hard
instance has realized density spread exactly $\gamma$. Since its size
$\delta$ may be chosen arbitrarily small, the lower bound of
Theorem~\ref{thm:small-good-lower} applies to the standard density-spread
parameter $\frac{\rho_{\max}}{\rho_{\min}}$,
which depends only on the extreme positive densities. The proof does
not establish the same lower bound under the stronger requirement that
a fixed positive fraction of the total size or value lie near both
density endpoints. Indeed, the logarithmic capacity obstruction is
generated by the ascending phases, whereas the calibration good serves
only to certify the realized spread. Requiring non-negligible endpoint
mass would therefore define a different model.

\subsection{The Limiting Density Frontier}

The upper and lower bounds developed in the previous two subsections
show that the approximation guarantee depends on the maximum item size
$\sigma_{\max}$. Our goal is to characterize the best deterministic
guarantee that remains achievable as the globally small-item assumption
becomes increasingly accurate, that is, as
$\sigma_{\max}\rightarrow0$. The appropriate notion of optimality is
therefore a limiting one: rather than asking for the best guarantee at
a fixed positive item size, we ask for the largest approximation factor
that can be achieved uniformly once the item-size bound is sufficiently
small.

\begin{definition}[Limiting supremum factor]
\label{def:limiting-supremum}
For a fixed model and $\sigma>0$, let $F(\sigma)$ be the
supremum of all factors $\alpha\in[0,1]$ for which some
deterministic online algorithm guarantees $\alpha$-FEF$1$ on
every instance in the model whose maximum item size is at most
$\sigma$.

If $0<\sigma_1\le\sigma_2$, then the size-$\sigma_1$ instance
class is contained in the size-$\sigma_2$ class. Hence $F(\sigma_1)\ge F(\sigma_2)$.
Since $F(\sigma)\in[0,1]$, the one-sided limit $\lim_{\sigma\rightarrow 0}F(\sigma)$ exists. We call this limit the \emph{limiting supremum factor}.
This terminology does not assert that the boundary factor is
attained for any fixed positive item-size bound.
\end{definition}

We now combine the matching upper and lower bounds to identify the
optimal deterministic approximation factor in the structured model.
Throughout this theorem, the algorithmic budget coincides with the
original unit budget. The approximation factor below is understood in the limiting-supremum
sense of Definition~\ref{def:limiting-supremum}; in particular, the
theorem does not assert that the boundary value is attained for any
fixed positive item-size bound.

\begin{theorem}[Limiting density frontier]
\label{thm:limiting-density-frontier}
Fix $n\ge 1$ and $\gamma>1$.  Consider the class
$\Icom$ of common-valuation instances with sizes common across
agents in the sense of Definition~\ref{def:common-uniform-sizes},
original unit budgets, under either the known-horizon or the
unknown-horizon model.
\begin{enumerate}[label=\textnormal{(\alph*)}]
    \item For every $\alpha\in
        \left(
            0,
            \frac1{1+\ln\gamma}
        \right)$,
    there exists $\sigma_0(\alpha,\gamma)>0$ such that
    Threshold$(\alpha)$ is prefix-wise
    $\alpha$-FEF1 against all recipients on every instance of
    $\Icom$ whose goods have size at most
    $\sigma_0(\alpha,\gamma)$.
    \item For every $\alpha\in
        \left(
            \frac1{1+\ln\gamma},
            1
        \right]$,
    and every $\sigma_{\max}>0$, no deterministic online algorithm
    guarantees $\alpha$-FEF1 on all instances in
    $\Icom$ whose goods have size at most
    $\sigma_{\max}$.  The hard instance supplied by Theorem~\ref{thm:small-good-lower} has realized spread
    exactly $\gamma$.
\end{enumerate}
The factor $\alpha=0$ is trivially achievable and is therefore
omitted. Consequently, for the function $F$ associated with this
model in Definition~\ref{def:limiting-supremum}, $\lim_{\sigma\rightarrow 0}F(\sigma)
= \frac{1}{1+\ln\gamma}$.
\end{theorem}

\begin{proof}
We first prove part \textnormal{(a)}. Define
\[
    \sigma_0(\alpha,\gamma)
    :=
    \frac{
        1-\alpha(1+\ln\gamma)
    }{
        2+\gamma
    }.
\]
This number is positive because $\alpha<\frac1{1+\ln\gamma}$.
If the input contains no goods, every fairness condition is trivial.
Assume henceforth that $G\ne\varnothing$.
Fix an instance whose goods all have size at most
$\sigma_0(\alpha,\gamma)$, and let $\bar\sigma
    :=
    \max_{g\in G}s(g)$.
Then $\bar\sigma\le\sigma_0(\alpha,\gamma)$.

Every good assigned to an agent has size at most $\bar\sigma$, so
Lemma~\ref{lem:capacity-no-fill}, applied with $B'=1$ and $\sigma_{\max}=\bar\sigma$, gives
\[
    s(A_i^t)+\bar\sigma\le1
    \quad
    \text{for every agent $i$ and prefix $t$}.
\]
Indeed,
\[
    \alpha(1+\ln\gamma)
    +(2+\gamma)\bar\sigma \le
    \alpha(1+\ln\gamma)
    +(2+\gamma)\sigma_0(\alpha,\gamma) =
    1.
\]
The preceding inequality is the key invariant of the upper-bound
analysis: every agent always retains enough residual capacity to
accept any future good satisfying the prescribed size bound. We use
this invariant to analyze charity and agent--agent comparisons
separately.

\paragraph{Comparisons with charity.}
Fix a prefix $t$, an agent $i$, and a feasible comparison set
$S\subseteq A_{\charity}^t$.  If $v(S)=0$, then $v(A_i)\ge0=\alpha v(S)$, so the desired inequality is immediate.  Assume $v(S)>0$, and let $\rho^* := \max_{g\in S:v(g)>0}\rho(g)$.
Choose $g^*\in S$ attaining $\rho^*$, and let
$t^*$ be its arrival round.  Since $g^*\in S$, it satisfies $s(g^*)\le1$ under the original
unit budget. The residual-capacity invariant also makes it feasible for every agent under the online budget. Since $g^*$ has positive value but was sent to charity, agent $i$ must therefore have failed the
threshold value condition: $v(A_i^{t^*-1}) \ge \alpha\rho^*$.
Agent values are nondecreasing over time, so at the currently
considered prefix,
\[
    v(A_i^t)
    \ge
    v(A_i^{t^*-1})
    \ge
    \alpha\rho^*.
\]
Moreover,
\[
    v(S)
    =
    \sum_{g\in S:v(g)>0}\rho(g)s(g)
    \le
    \rho^* s(S)
    \le
    \rho^*.
\]
Therefore $v(A_i^t)\ge\alpha v(S)$.

\paragraph{Agent--agent comparisons.}
The residual-capacity invariant established above implies, at every prefix, $s(A_i^t)+\bar\sigma\le1$ for every $i$.
Since every arriving good has size at most $\bar\sigma$,
condition~\eqref{eq:universal-arrival-feasibility} holds at every
assignment. Lemma~\ref{lem:balanced-agent-envy} therefore implies
exact agent--agent FEF$1$.
Both hold at every prefix, so Threshold$(\alpha)$ is prefix-wise
$\alpha$-FEF$1$ against all recipients on $\Icom$.

It remains to prove part~\textup{(b)}.
Applying Theorem~\ref{thm:small-good-lower} with
$\varepsilon=0$, its hypothesis holds because
$\alpha>1/(1+\ln\gamma)$.
Hence, for every deterministic online algorithm $\mathcal A$ and every
$\sigma_{\max}>0$, Theorem~\ref{thm:small-good-lower} supplies a
fixed finite instance $\mathcal H_{\mathcal A}\in\Icom$ whose goods
all have size at most $\sigma_{\max}$ and on which $\mathcal A$
fails to achieve $\alpha$-FEF$1$.
\end{proof}

As in Definition~\ref{def:limiting-supremum}, no claim in Theorem \ref{thm:limiting-density-frontier} is made that the
boundary factor is attained for any fixed positive item-size bound. The lower bound in part \textnormal{(b)} has quantifier order $\forall\mathcal A\ \exists\mathcal H_{\mathcal A}$,
where $\mathcal H_{\mathcal A}$ is a fixed finite input sequence
selected before the actual execution of $\mathcal A$.  Thus, the theorem
does not claim that one universal sequence defeats every deterministic
algorithm simultaneously.

\subsection{The Limiting Augmentation Frontier}

The previous theorem characterizes the best limiting factor as the maximum good size tends to zero without
resource augmentation. We now allow the algorithm to use an
algorithmic budget $1+\varepsilon$ while fairness continues to be
evaluated with respect to the original budgets. The resulting
best limiting factor as the maximum good size tends to zero shifts upward by exactly the augmentation factor.
\begin{theorem}[Limiting augmentation frontier]
\label{thm:limiting-augmentation-frontier}

Fix $n\ge1$, $\gamma>1$, and
$\varepsilon\ge0$, and let $\alpha^*(\varepsilon)
=
\min\left\{
1,
\frac{1+\varepsilon}{1+\ln\gamma}
\right\}$.

Under either the known-horizon or unknown-horizon model, the limiting
supremum deterministic approximation factor for the structured class
$\Icom$, under algorithmic budget $1+\varepsilon$ and fairness
evaluated with respect to the original unit budgets, is
$\alpha^*(\varepsilon)$.

More precisely,

\begin{enumerate}[label=\textnormal{(\alph*)}]
    \item For every $\alpha\in(0,1]$ satisfying
    \begin{equation}
    \label{eq:augmentation-strict-upper}
        \alpha(1+\ln\gamma)<1+\varepsilon,
    \end{equation}
    let $\sigma_0(\alpha,\varepsilon,\gamma)
        :=
        \frac{
            (1+\varepsilon)-\alpha(1+\ln\gamma)
        }{
            2+\gamma
        }$.
    Then $\sigma_0>0$, and Threshold$(\alpha)$, using algorithmic
    budget $1+\varepsilon$, is prefix-wise $\alpha$-FEF1 against all
    recipients with respect to the original unit budgets whenever every
    good has size at most $\sigma_0$.
    \item
    For every $\alpha\in(0,1]$ satisfying $\alpha(1+\ln\gamma)>1+\varepsilon$, and for every $\sigma_{\max}>0$, no deterministic online algorithm that
    maintains $s(A_i^t)\le1+\varepsilon$ for every agent $i$ and prefix $t$ guarantees $\alpha$-FEF1 on all instances in
    $\Icom$ whose goods have size at most
    $\sigma_{\max}$.  The hard instance may be chosen to have realized
    spread exactly $\gamma$.
    \item
    In particular, if $\varepsilon>\ln\gamma$,
    then part \textnormal{(a)} applies with $\alpha=1$ and $\sigma_0
        =
        \frac{\varepsilon-\ln\gamma}{2+\gamma}>0$.
    Thus, Threshold$(1)$ is prefix-wise exactly FEF1 against all
    recipients for sufficiently small goods.
    \item
    If $\varepsilon<\ln\gamma$, then part \textnormal{(b)} applies with $\alpha=1$, and exact FEF1 is impossible for every positive item-size bound.
\end{enumerate}
No claim is made regarding boundary attainment when $\varepsilon=\ln\gamma$.
\end{theorem}

\begin{proof}
We first establish part \textnormal{(a)}. Since $\alpha(1+\ln\gamma)<1+\varepsilon$, we have that
$\sigma_0
    :=
    \frac{
        (1+\varepsilon)-\alpha(1+\ln\gamma)
    }{
        2+\gamma
    }$
is strictly positive.
If the input contains no goods, every fairness condition is trivial.
Assume henceforth that $G\ne\varnothing$.
Fix an instance whose goods all have size at most $\sigma_0$, and
define $\bar\sigma := \max_{g\in G}s(g)$.
Then $\bar\sigma\le\sigma_0$.
Applying Lemma~\ref{lem:capacity-no-fill} with $B'=1+\varepsilon$ and $\sigma_{\max}=\bar\sigma$ gives us $s(A_i^t)+\bar\sigma \le 1+\varepsilon$ for every agent $i$ and prefix $t$, because
\begin{equation*}
    \alpha(1+\ln\gamma)
    +(2+\gamma)\bar\sigma \le
    \alpha(1+\ln\gamma)
    +(2+\gamma)\sigma_0 =
    1+\varepsilon.
\end{equation*}
As in the proof of Theorem~\ref{thm:limiting-density-frontier}, this
residual-capacity invariant allows us to analyze charity and
agent--agent comparisons separately.
Fix an agent $i$ and a feasible charity set $S\subseteq A_{\charity}^t$.
If $v(S)=0$, then $v(A_i)\ge0=\alpha v(S)$.
Assume $v(S)>0$, and let $g^*\in S$ maximize the common
density $\rho(g)$ over positive-value goods in $S$; let
$t^*$ be its arrival round.  Since $g^*\in S$, it satisfies $s(g^*)\le1$ under the original
unit budget. The residual-capacity invariant also makes it feasible for
every agent under the online budget. Since $g^*$ has positive value but
was sent to charity, agent $i$ must therefore have failed the
threshold value condition: $ v(A_i^{t^*-1})
    \ge
    \alpha\rho(g^*)$.
By monotonicity of accepted value,
\[
    v(A_i)
    \ge
    v(A_i^{t^*-1})
    \ge
    \alpha\rho(g^*).
\]
Also,
\[
    v(S) =
    \sum_{g\in S:v(g)>0}\rho(g)s(g) \le
    \rho(g^*)s(S) \le
    \rho(g^*),
\]
because $S$ is feasible for the original unit budget.  Therefore $v(A_i)\ge\alpha v(S)$.
\paragraph{Agent--agent comparisons.}
At every prefix,
\[
    s(A_i^t)+\bar\sigma\le1+\varepsilon
    \quad
    \text{for every }i.
\]
Every arriving good has size at most $\bar\sigma$, so it is feasible
for every agent immediately before any assignment.  Lemma
\ref{lem:balanced-agent-envy} gives exact agent--agent FEF1.
Since $\alpha\le1$, the
allocation is prefix-wise $\alpha$-FEF1 against all recipients.

The lower bound in part \textnormal{(b)} follows directly from
Theorem~\ref{thm:small-good-lower}. Indeed, the hypothesis of that
theorem is precisely
\[
\alpha(1+\ln\gamma)>1+\varepsilon.
\]

Parts \textnormal{(c)} and \textnormal{(d)} are obtained by substituting
$\alpha=1$.

It remains to identify the best limiting factor as the maximum good size tends to zero.
Let
\[
c
=
\alpha^*(\varepsilon)
=
\min\left\{
1,
\frac{1+\varepsilon}{1+\ln\gamma}
\right\}.
\]

Suppose first that $c<1$. Then $c=\frac{1+\varepsilon}{1+\ln\gamma}$.
The upper bound in part \textnormal{(a)} established above yields every factor
$\alpha<c$ for sufficiently small goods, whereas the lower bound in part
\textnormal{(b)} excludes every factor $\beta>c$ for every positive
item-size bound. Consequently,
$\lim_{\sigma\rightarrow 0}F(\sigma)=c$, so $c$ is the limiting supremum
approximation factor.

Now suppose that $c=1$. The upper bound in part \textnormal{(a)} yields every factor
$\alpha<1$ for sufficiently small goods, while no approximation
factor can exceed $1$. Hence, the limiting supremum is $1$.

Moreover, if
$\varepsilon>\ln\gamma$, then the upper-bound construction with
$\alpha=1$ (part \textnormal{(c)}) gives exact FEF$1$ for sufficiently small goods.
At the boundary $\varepsilon=\ln\gamma$, the limiting supremum remains $1$, although no claim is made that
exact FEF$1$ is attained for any fixed positive item-size bound.

Therefore, $\alpha^*(\varepsilon)
=
\min\left\{
1,
\frac{1+\varepsilon}{1+\ln\gamma}
\right\}$ is the limiting supremum deterministic approximation factor.
Moreover, $\ln\gamma$ is the infimum augmentation threshold:
every $\varepsilon>\ln\gamma$ guarantees exact FEF$1$ for all
sufficiently small goods, every $\varepsilon<\ln\gamma$ does not,
and attainment at the boundary $\varepsilon=\ln\gamma$ remains
open.
\end{proof}

The frontier theorems above rely critically on the assumption that
item sizes are common across agents. The charity-side argument extends
to agent-specific sizes, since the residual-capacity analysis of
Lemma~\ref{lem:capacity-no-fill} is carried out independently for each
agent and requires only the corresponding small-item condition.
However, the exact agent--agent FEF$1$ argument does not extend
directly. When a good is assigned to agent $j$, its density and
feasibility threshold for another agent $i$ may differ, so the
balanced-acceptance argument of
Lemma~\ref{lem:balanced-agent-envy} no longer applies.

Consequently,
Theorems~\ref{thm:limiting-density-frontier} and
\ref{thm:limiting-augmentation-frontier} establish the limiting
frontiers only for the model with sizes common across agents. The ascending
construction still yields the charity-side lower-bound obstruction $\alpha(1+\ln\gamma)\le1+\varepsilon$
under the structured common-across-agent-size assumptions, but the
corresponding upper bound for fully agent-specific sizes remains open.

\begin{remark}[Only individually feasible goods need be small]
Since Threshold$(\alpha)$ rejects every good with $s(g)>1$, the
upper-bound proofs of Theorems~\ref{thm:limiting-density-frontier}
and~\ref{thm:limiting-augmentation-frontier} use the small-item bound
only for goods that are individually feasible under the original unit
budget. Goods with $s(g)>1$ are never assigned to an agent and cannot
belong to an original-budget feasible comparison set. The lower bound in
Theorem~\ref{thm:small-good-lower} already uses globally small goods.
Hence both frontier theorems remain valid, under known and unknown
horizons, if the size bound is imposed only on individually feasible
goods.
\end{remark}

% % %====================================================================

\section{Why Small Items Do Not Make Aug-GreedyFit Exact}
\label{app:counter}
%====================================================================

It is tempting to conjecture that \emph{small items alone} suffice for exact
FEF$1$ under augmentation: "if $\size_i(g)\le\eps\bud_i$ for all $i,g$, then
\textsc{Aug-GreedyFit}$(\eps)$ is exactly FEF$1$." This is \emph{false} when the
density spread exceeds $1$; \Cref{thm:aug-ub} correctly gives only
$\min\{1,(1+\eps)/\spread\}$, and the following instance shows the factor cannot be
improved to $1$ merely by shrinking items.

\begin{example}[Arbitrarily small goods do not make
\textsc{Aug-GreedyFit} exact when \(\Gamma>1\)]
Take one agent with original budget \(B=1\) and online budget
\(1.1\). Fix any \(\sigma>0\), and choose a multiple \(Q\) of \(10\)
such that \(Q\ge10\) and \(1/Q\le\sigma\). Every good has size
\(1/Q\).

First, \(11Q/10\) low goods arrive, each having value \(1\). They have
density \(Q\). The algorithm accepts all of them, filling the online
budget exactly:
\[
\frac{11Q}{10}\cdot\frac1Q=1.1.
\]
Next, \(Q\) high goods arrive, each having value \(2\) and density
\(2Q\). The agent has no remaining online capacity, so all high goods
are assigned to charity. The realized density spread is therefore
\(\Gamma=2\).

The agent's value is \(11Q/10\). Let \(S\) consist of the \(Q\) high
goods in charity. Then \(s(S)=1\), so \(S\) is feasible under the
original budget, and for every \(x\in S\),
\[
v(S\setminus\{x\})=2(Q-1).
\]
Since \(Q\ge10\),
\[
\frac{11Q}{10}<2(Q-1),
\]
and hence the allocation is not exact FEF1. Because \(1/Q\le\sigma\)
and \(\sigma>0\) was arbitrary, the failure occurs for arbitrarily
small goods.
\end{example}

The example shows that merely imposing a small-item condition does not
upgrade \textsc{Aug-GreedyFit}'s arbitrary-item-size guarantee to
exact FEF1.  For that particular greedy algorithm, the proven general
factor remains $\min\left\{ 1, \frac{1+\varepsilon}{\Gamma}
    \right\}$, and exactness follows from this analysis only when $\varepsilon\ge\Gamma-1$.
This does not contradict the sharper threshold policy of
Section~\ref{sec:sharp-frontier}.  In the more structured
common-valuation, common-across-agent-size, globally-small model,
Threshold$(1)$ achieves exact FEF1 for every $\varepsilon>\ln\gamma$.
The example therefore separates the broad arbitrary-size greedy
analysis from the sharper structured small-item frontier; it is not a
lower bound against all online algorithms in the latter model.

\section{Density-Spread Lower Bound for Arbitrary Item Sizes}
\label{sec:aux-density}
%====================================================================
For arbitrary item sizes, the frontier characterizations of
Theorems~\ref{thm:limiting-density-frontier}
and~\ref{thm:limiting-augmentation-frontier} no longer apply. This appendix
records a complementary lower bound showing that, under the paper's
density-spread convention, no deterministic online algorithm can
guarantee an approximation factor exceeding
$\gamma/(2\gamma-1)$, even with a known horizon.

Unlike the frontier results, this lower bound does not characterize the
optimal approximation factor. Instead, it provides an auxiliary
impossibility result for the arbitrary-item-size setting. The
construction uses a high-density marker good in Continuation~A to make
the realized density spread exactly $\gamma$; all positive-valued goods
remain individually feasible.

\begin{proposition}[Density-spread lower bound for arbitrary item sizes]
\label{prop:arbitrary-size-density-lower}
Fix an integer $k\ge 1$, a density bound $\gamma>1$, and $\alpha>\frac{\gamma}{2\gamma-1}$.
Even with a known horizon, no deterministic online algorithm guarantees $\alpha$-FEF$k$ on all instances with two agents, equal unit budgets, common valuations, uniform item sizes, and realized density spread exactly $\gamma$.
\end{proposition}

\begin{proof}
Choose an integer $Q>k$ sufficiently large that
\begin{equation}
\alpha(Q-k)>\gamma
\quad\text{and}\quad
\gamma Q+\gamma(\gamma-1)
<
\alpha(2\gamma-1)(Q-k).
\label{eq:arbitrary-size-Q}
\end{equation}
Such a choice is possible because
$\alpha>\gamma/(2\gamma-1)$.

Every good has size $1/Q$, and both agents have budget $1$. Define three
kinds of goods:
\begin{itemize}
    \item a low good has value $1$ and density $Q$;
    \item a high good has value $\gamma$ and density $\gamma Q$;
    \item a zero good has value $0$.
\end{itemize}

The known horizon is $6Q$. The first $3Q$ arrivals are low goods. Since
the two agents can hold at most $2Q$ goods in total, at least $Q$ low
goods are assigned to charity during this prefix.

Consider the following two continuations.

\paragraph{Continuation A.}
The remaining arrivals are one high good followed by $3Q-1$ zero goods.

Let $m_i$ be the number of low goods assigned to agent $i$ during the
common prefix. The final value of agent $i$ in this continuation is at most
$m_i+\gamma$, because there is only one high good.

Let $S$ be any set of $Q$ low goods in charity. For every
$X\subseteq S$ with $|X|\le k$, $v(S\setminus X)\ge Q-k$.
Therefore, if the final allocation is $\alpha$-FEF$k$, then $m_i+\gamma\ge \alpha(Q-k)$,
and hence
\begin{equation}
m_i\ge \alpha(Q-k)-\gamma
\quad\text{for each }i\in\{1,2\}.
\label{eq:low-prefix-count}
\end{equation}

\paragraph{Continuation B.}
The remaining $3Q$ arrivals are high goods.

The two continuations have the same announced horizon and the same first
$3Q$ arrivals. Since the algorithm is deterministic, the values $m_1,m_2$
after the common prefix are the same in both continuations.

Agent $i$ can receive at most $Q-m_i$ high goods. Thus,
\begin{equation*}
v(A_i) \le m_i+\gamma(Q-m_i)  = \gamma Q-(\gamma-1)m_i \le \gamma Q-(\gamma-1)(\alpha(Q-k)-\gamma)  < \alpha\gamma(Q-k),
\end{equation*}
where the last inequality follows from
\eqref{eq:arbitrary-size-Q}.

The number of high goods assigned to charity is at least
\begin{equation*}
    3Q-\sum_{i=1}^2(Q-m_i) =Q+m_1+m_2 \ge Q+2(\alpha(Q-k)-\gamma) \ge Q,
\end{equation*}
where the final inequality follows from
\eqref{eq:arbitrary-size-Q}.

Let $S$ be any set of $Q$ high goods in charity. Then, for every
$X\subseteq S$ with $|X|\le k$,
\[
v(S\setminus X)\ge \gamma(Q-k).
\]
Since $v(A_i)<\alpha\gamma(Q-k)$, both agents violate
$\alpha$-FEF$k$ toward charity.

If the algorithm fails on Continuation A, that continuation is the required
instance. Otherwise, the preceding argument shows that it fails on
Continuation B. Both continuations have positive densities $Q$ and
$\gamma Q$, and all positive-valued goods are individually feasible.
Therefore, their realized density spread under
Definition~\ref{def:density-spread} is exactly $\gamma$.
\end{proof}

All positive-valued goods in both continuations have size
$1/Q\le1$ and are therefore included in the density-spread parameter of
Definition~\ref{def:density-spread}. The single high good in
Continuation~A serves only to ensure that this continuation, like
Continuation~B, has realized spread exactly $\gamma$.

\section{An Ascending-Capacity Optimization Lemma}
\label{app:lp}
%====================================================================

This appendix proves a technical optimization lemma used in the
lower-bound proof of \Cref{thm:small-good-lower}. That proof reduces
to minimizing the total mass assigned to an increasing sequence of
density levels subject to ascending cumulative-capacity constraints.
The lemma solves this optimization problem in closed form, allowing the
same argument to be invoked without reproving the underlying
linear-programming calculation.

\begin{lemma}[Ascending-capacity optimization]
\label{lem:ascending-capacity}
Let $0<h_0<h_1<\cdots<h_M$ and $0\le c_0\le c_1\le\cdots\le c_M$.
Among all vectors $ x_0,\dots,x_M\ge0$
satisfying $W_j
    := \sum_{\ell=0}^j h_\ell x_\ell
    \ge c_j$ for all $j\in \{0,\dots,M\}$, the minimum of $\sum_{\ell=0}^M x_\ell$ is
\begin{equation}
\label{eq:ascending-lp-optimum}
    \frac{c_0}{h_0}
    +
    \sum_{\ell=1}^M
    \frac{c_\ell-c_{\ell-1}}{h_\ell}.
\end{equation}
It is attained by $x_0=\frac{c_0}{h_0},
    \quad
    x_\ell
    =
    \frac{c_\ell-c_{\ell-1}}{h_\ell}
    \quad(\ell=1,\dots,M)$, 
for which $W_j=c_j
    \quad(j=0,\dots,M)$.
\end{lemma}
\begin{proof}
We first verify that the proposed solution is feasible and then prove
that no feasible solution can have a smaller objective value.
Since $0\le c_0\le\cdots\le c_M$,
the proposed vector is nonnegative. Moreover,
\[
W_j =
h_0\frac{c_0}{h_0}
+
\sum_{\ell=1}^{j}
h_\ell
\frac{c_\ell-c_{\ell-1}}{h_\ell} =
c_0+
\sum_{\ell=1}^{j}
(c_\ell-c_{\ell-1}) = c_j.
\]
Thus, every constraint is satisfied with equality, and the objective
value is exactly \eqref{eq:ascending-lp-optimum}.

It remains to show that no feasible solution achieves a smaller
objective value. Let $\mathbf x=(x_0,\dots,x_M)$
be any feasible solution, and define $W_{-1}:=0$.
Since $W_\ell-W_{\ell-1}=h_\ell x_\ell,$
we may rewrite the variables as $x_\ell=\frac{W_\ell-W_{\ell-1}}{h_\ell}$.
Substituting into the objective and applying summation by parts gives
\begin{equation}
\sum_{\ell=0}^{M}x_\ell
=
\sum_{\ell=0}^{M}
\frac{W_\ell-W_{\ell-1}}{h_\ell} =
\sum_{\ell=0}^{M-1}
W_\ell
\left(
\frac1{h_\ell}
-
\frac1{h_{\ell+1}}
\right)
+
\frac{W_M}{h_M}.
\label{eq:ascending-abel}
\end{equation}

Since $h_0<h_1<\cdots<h_M$,
every coefficient on the right-hand side is strictly positive.
Moreover, feasibility implies
\[
W_\ell\ge c_\ell,
\quad
\ell=0,\dots,M.
\]
Replacing each $W_\ell$ by its lower bound therefore yields
\begin{equation*}
\sum_{\ell=0}^{M}x_\ell \ge
\sum_{\ell=0}^{M-1}
c_\ell
\left(
\frac1{h_\ell}
-
\frac1{h_{\ell+1}}
\right)
+
\frac{c_M}{h_M} =
\frac{c_0}{h_0}
+
\sum_{\ell=1}^{M}
\frac{c_\ell-c_{\ell-1}}{h_\ell},
\end{equation*}
which is precisely the objective value attained by the proposed
solution. Hence, the proposed solution is optimal.
\end{proof}

\begin{corollary}[Specialization used in the lower-bound proofs]
\label{cor:ascending-lp}
Let $h_j=\gamma^{j/M}$ and $c_j=\alpha h_j(1-2\delta)-d$,
where $0\le d<\alpha(1-2\delta)$.
Then, $c_0=\alpha(1-2\delta)-d>0$,
so $(c_j)$ is nonnegative and nondecreasing. Consequently,
\[
\sum_{\ell=0}^M x_\ell
\ge
\alpha(1-2\delta)(1+S_M)-d,
\]
where $S_M
=
M(1-\gamma^{-1/M})
\longrightarrow
\ln\gamma$.
\end{corollary}

\begin{proof}
Applying \Cref{lem:ascending-capacity} gives
\begin{equation*}
\sum_{\ell=0}^M x_\ell \ge
\frac{c_0}{h_0}
+
\sum_{\ell=1}^M
\frac{c_\ell-c_{\ell-1}}{h_\ell} =
\alpha(1-2\delta)-d
+
\alpha(1-2\delta)
\sum_{\ell=1}^M
\left(
1-\frac{h_{\ell-1}}{h_\ell}
\right) =
\alpha(1-2\delta)(1+S_M)-d.
\end{equation*}
\end{proof}

Theorem~\ref{thm:small-good-lower} applies this corollary with
$d=\gamma y_i$, where $y_i\in\{0,\delta\}$.
Condition~\eqref{eq:small-positive-target} ensures
$d<\alpha(1-2\delta)$.

\section{Omitted Proofs for Section~\ref{sec:la}}
\label{app:omitted:la}

\subsection{Proof of Lemma~\ref{lem:reserve}}
\label{app:proof:lem:reserve}
\ReserveFeasibility*

\begin{proof}
We prove that the invariant $\size_i(A_i)+\mathrm{res}_i\le\bud_i$ holds for every agent $i$
throughout the execution. The claimed
feasibility properties then follow immediately. Initially, $A_i=\varnothing$ and $\mathrm{res}_i=\sum_\theta
q_{i,\theta}\size_{i,\theta}=\size_i(\widehat A_i)\le\bud_i$, the last inequality
because $\widehat A$ is feasible for the planner. Consider the three assignment
branches. Suppose the invariant holds immediately before processing an arriving
good. We verify that it is preserved in each branch of
\textsc{Plan-Reserve-Fulfill}.

\paragraph{Filling a planned agent quota.} If a type-$\theta$ good fills an unfilled quota of
agent $i$, then $q^{\mathrm{unf}}_{i,\theta}>0$ before the step, so $\mathrm{res}_i$
includes a term $\ge\size_{i,\theta}$; thus
$\size_i(A_i)+\size_{i,\theta}\le\size_i(A_i)+\mathrm{res}_i\le\bud_i$, so the assignment
is feasible. After it, $\size_i(A_i)$ increases by $\size_{i,\theta}$ and $\mathrm{res}_i$
decreases by exactly $\size_{i,\theta}$, leaving $\size_i(A_i)+\mathrm{res}_i$ unchanged.

\paragraph{Assigning to charity.}Sending a good to charity changes no agent's $A_i$ and does not
increase any $\mathrm{res}_i$ (a charity quota decrement leaves $\mathrm{res}_i$ untouched), so the
invariant is preserved.

\paragraph{Assigning an excess good to agent $i$.} This branch is taken only when
$\size_i(A_i)+\size_{i,\theta}+\mathrm{res}_i\le\bud_i$ is checked explicitly; afterwards
$\size_i(A_i)$ grows by $\size_{i,\theta}$ and $\mathrm{res}_i$ is unchanged, so the
invariant holds. In all cases $\size_i(A_i)\le\size_i(A_i)+\mathrm{res}_i\le\bud_i$ at
termination.

Thus, the invariant is preserved throughout the execution. At
termination,
\[
\size_i(A_i)\le
\size_i(A_i)+\mathrm{res}_i\le\bud_i
\]
for every agent, so the returned allocation is budget feasible.
Moreover, every good used to satisfy a planned agent quota is accepted
at the moment it arrives without violating feasibility.
\end{proof}

\subsection{Proof of Theorem~\ref{thm:consistency}}
\label{app:proof:thm:consistency}
\PerfectPredictionConsistency*

\begin{proof}
Suppose $f=\widehat f$. Then, for every type $\theta$, exactly
\[
f_\theta=\widehat f_\theta=\sum_{r\in\Rec} q_{r,\theta}
\]
goods of type $\theta$ arrive, matching the total planned quota of that type.

By \Cref{lem:reserve}, every planned agent quota can be filled
immediately upon the arrival of a matching good without violating
budget feasibility. Moreover,
\textsc{Plan-Reserve-Fulfill} always gives priority to unfilled
planned quotas, first for agents and then for charity, before treating
a good as excess. Since the total number of arriving type-$\theta$
goods exactly equals the total planned quota of that type, every
planned quota is filled, and no good is ever processed as an excess
good.

It follows that, for every recipient $r$ and type $\theta$, the final
allocation $A_r$ contains exactly $q_{r,\theta}$ goods of type $\theta$.
Hence, the online allocation realizes the planned allocation $\widehat A$ up
to a relabelling of identical goods.
\end{proof}

\subsection{Proof of Corollary~\ref{cor:fefx-lift}}
\label{app:proof:cor:fefx-lift}

By \Cref{thm:consistency}, every
fairness property of $\widehat A$ determined solely by the per-recipient
type counts is inherited by the online allocation, including FEF,
FEF$k$, $\alpha$-FEF$k$, additive $\eta$-FEF$k$, and FEFx.
\begin{restatable}[Lifting offline fairness guarantees]{corollary}{LiftingOfflineFairness}
\label{cor:fefx-lift}
Suppose that the offline planner returns an additive
$\eta_0$-FEF$k$ allocation on the predicted multiset
$\widehat M$. If $f=\widehat f$, then
\textsc{Plan-Reserve-Fulfill} returns an additive
$\eta_0$-FEF$k$ allocation.

In particular, an exact offline FEFx planner gives exact FEFx under
correct predictions, and hence also exact FEF1.
\end{restatable}
\begin{proof}
The first claim follows from \Cref{thm:consistency} because additive
$\eta_0$-FEF$k$ depends only on the type counts assigned to each
recipient. The second claim follows by applying the same argument to
an offline FEFx allocation. The implication from FEFx to FEF1 follows
by definition.
\end{proof}

\subsection{Characterization of $k$-Violation}
\label{app:k-violation}

The smallest additive slack required by an allocation $A$ can be expressed
through the \emph{$k$-violation}, defined as
\[
\operatorname{viol}_k(A)
:=
\max_{\substack{
i\in N,\ r\in R,\\
S\subseteq A_r,\ s_i(S)\le B_i
}}
\left[
\min_{\substack{X\subseteq S\\|X|\le k}}
v_i(S\setminus X)
-
v_i(A_i)
\right]_+,
\]
where $[x]_+:=\max\{x,0\}$.
Equivalently, $A$ is additive $\eta$-FEF$k$ if and only if $\operatorname{viol}_k(A)\le\eta$. In particular, we next prove that the above explicit expression is equivalent to the one specified in \Cref{sec:Robustness}.

\begin{lemma}[Characterization of $k$-violation]
\label{lem:viol-characterization}
For every allocation $A$ and every $\eta\ge 0$, $A$ is additive
$\eta$-FEF$k$ if and only if $\operatorname{viol}_k(A)\le \eta$.
Equivalently,
\[
\operatorname{viol}_k(A)
=
\min\{\eta\ge 0 : A \text{ is additive } \eta\text{-FEF}k\}.
\]
\end{lemma}

\begin{proof}
By Definition~\ref{def:etafefk}, \(A\) is additive \(\eta\)-FEF\(k\) if and only if,
for every agent \(i\in N\), recipient \(r\in R\), and set
\(S\subseteq A_r\) satisfying \(s_i(S)\le B_i\), there exists
\(X\subseteq S\) with \(|X|\le k\) such that
\[
v_i(S\setminus X)-v_i(A_i)\le \eta.
\]
For a fixed admissible triple \((i,r,S)\), this is equivalent to
\[
\min_{\substack{X\subseteq S\\ |X|\le k}}
\bigl(v_i(S\setminus X)-v_i(A_i)\bigr)
\le \eta.
\]
Define
\[
M(A):=
\max_{\substack{
i\in N,\ r\in R,\ S\subseteq A_r\\
s_i(S)\le B_i
}}
\;
\min_{\substack{X\subseteq S\\ |X|\le k}}
\bigl(v_i(S\setminus X)-v_i(A_i)\bigr).
\]
It follows that \(A\) is additive \(\eta\)-FEF\(k\) if and only if
\(M(A)\le\eta\). Since \(\eta\ge 0\),
\[
M(A)\le\eta
\quad\Longleftrightarrow\quad
[M(A)]_+\le\eta.
\]
By definition, \(\operatorname{viol}_k(A)=[M(A)]_+\). Therefore,
\[
A\text{ is additive }\eta\text{-FEF}k
\quad\Longleftrightarrow\quad
\operatorname{viol}_k(A)\le\eta.
\]
Taking the minimum over \(\eta\ge0\) proves
\[
\operatorname{viol}_k(A)
=
\min\{\eta\ge0:A\text{ is additive }\eta\text{-FEF}k\}. \qedhere
\]
\end{proof}

\subsection{Robustness via Local Prediction Errors}
\label{app:proof:thm:robust}

\begin{restatable}{theorem}{RobustnessLocal}
\label{thm:robust}
Suppose the planned allocation $\widehat A$ is additive
$\eta_0$-FEF$k$. Let $A$ be the allocation returned by
\textsc{Plan-Reserve-Fulfill}. Let $m_i$ denote the total
$v_i$-value of unfilled planned quotas for agent $i$ and
$e_{i,r}$ the total $v_i$-value of excess realized
goods assigned to recipient $r$. Then, $\viol_k(A)\le
\eta_0+\max_{i,r}(m_i+e_{i,r})$.
\end{restatable}

\begin{proof}
Fix an agent $i$, a recipient $r$, and a set $S\subseteq A_r$
that is feasible for $i$. Partition $S=S^{\mathrm{pl}}\uplus S^{\mathrm{ex}}$, where $S^{\mathrm{pl}}$ contains the realized goods that filled
planned quotas of $r$, and $S^{\mathrm{ex}}$ contains the excess
goods assigned to $r$.

For each good in $S^{\mathrm{pl}}$, choose the distinct predicted
copy whose quota it filled. Let
$\widehat S^{\mathrm{pl}}\subseteq\widehat A_r$ be the set of these
predicted copies. Each realized good and its corresponding predicted
copy have the same type. Therefore,
\[
s_i(\widehat S^{\mathrm{pl}})
=s_i(S^{\mathrm{pl}})
\le s_i(S)
\le B_i
\]
and
\[
v_i(\widehat S^{\mathrm{pl}})
=v_i(S^{\mathrm{pl}}).
\]

Since $\widehat A$ is additive $\eta_0$-FEF$k$, there exists
$\widehat X\subseteq\widehat S^{\mathrm{pl}}$ with
$|\widehat X|\le k$ such that
\[
v_i(\widehat A_i)
\ge
v_i(\widehat S^{\mathrm{pl}}\setminus\widehat X)-\eta_0.
\]
Let $X\subseteq S^{\mathrm{pl}}$ be the realized goods corresponding
to the predicted copies in $\widehat X$. Then $|X|\le k$ and
\[
v_i(\widehat S^{\mathrm{pl}}\setminus\widehat X)
=
v_i(S^{\mathrm{pl}}\setminus X).
\]

The filled planned quotas of agent $i$ contribute exactly the value
of the corresponding predicted copies. The only planned value missing
from $A_i$ is counted by $m_i$, while excess goods assigned to
$i$ have nonnegative value. Hence,
\[
v_i(A_i)\ge v_i(\widehat A_i)-m_i.
\]
Combining the preceding inequalities gives
\[
v_i(A_i)
\ge
v_i(S^{\mathrm{pl}}\setminus X)-\eta_0-m_i.
\]

Finally, $v_i(S^{\mathrm{ex}})\le e_{i,r}$, and therefore
\[
v_i(S\setminus X) =
v_i(S^{\mathrm{pl}}\setminus X)+v_i(S^{\mathrm{ex}}) \le
v_i(A_i)+\eta_0+m_i+e_{i,r}.
\]
Rearranging proves
\[
v_i(A_i)
\ge
v_i(S\setminus X)-\eta_0-m_i-e_{i,r}.
\]
Taking the maximum over all feasible comparisons proves the stated
bound on $\viol_k(A)$.
\end{proof}

\subsection{Proof of Theorem~\ref{thm:robust-global}}
\label{app:proof:thm:robust-global}
\PredictionErrorRobustness*

\begin{proof}
We bound the two local error terms appearing in
\Cref{thm:robust} by the corresponding prediction error.
For each type $\theta$, an unfilled planned quota can arise only if
fewer than $\widehat f_\theta$ goods of that type are realized. Hence, the
number of unfilled type-$\theta$ quotas is at most
$(\widehat f_\theta-f_\theta)_+$.
Similarly, an excess good of type $\theta$ can arise only if more than
$\widehat f_\theta$ goods of that type are realized. Thus, the number of
excess type-$\theta$ goods is at most
$(f_\theta-\widehat f_\theta)_+$.

Consequently, for every agent $i$ and recipient $r$,
\[
m_i\le\sum_{\theta}\val_{i,\theta}(\widehat f_\theta-f_\theta)_+,
\quad
e_{i,r}\le\sum_{\theta}\val_{i,\theta}(f_\theta-\widehat f_\theta)_+.
\]
For every type $\theta$, at most one of
$(\widehat f_\theta-f_\theta)_+$ and
$(f_\theta-\widehat f_\theta)_+$ is nonzero, and their sum equals
$|f_\theta-\widehat f_\theta|$. Therefore,
\[
m_i+e_{i,r}\ \le\ \sum_\theta\val_{i,\theta}\,|f_\theta-\widehat f_\theta|\ =\ \Delta_i(f,\widehat f)\ \le\ \Delta_\infty(f,\widehat f).
\]
Substituting this bound into \Cref{thm:robust} gives $\viol_k(A)\le
\eta_0+\Delta_\infty(f,\widehat f)$,
proving the first claim. The common-valuation specialization follows
immediately from the definition of $\Delta(f,\widehat f)$.
\end{proof}

\subsection{Proof of Theorem~\ref{thm:pred-lb}}
\label{app:proof:thm:pred-lb}
\PredictionErrorLowerBound*

\begin{proof}
We construct two instances sharing the same prediction and the same
arrival prefix. The first instance satisfies the prediction exactly and
forces every perfectly consistent algorithm to reserve almost all
capacity for low-valued goods. The second instance differs only in its
final arrivals, replacing the predicted zero-valued goods by
high-valued goods. Since the algorithm cannot distinguish the two
instances before this point, it has too little remaining capacity to
accept many of the high-valued goods.

Formally, let $\bud_1=\bud_2=1$. Choose $Q$ sufficiently large that
\[
\frac{Q-2k}{Q+2k}>1-\frac{\zeta}{2},
\]
and choose $\delta>0$ to be small, as fixed later. Every good has size $1/Q$.

The prediction consists of $3Q$ \emph{low} goods of value $\delta$ and size $1/Q$, and
$Q+2k$ \emph{zero-padding} goods of value $0$ and size $1/Q$. Thus, the predicted
horizon is $\horizon=4Q+2k$.

\paragraph{Perfect-prediction instance.}
Assume the prediction is correct and the first $3Q$ arrivals are the low
goods. After they arrive, since the two agents together can hold at most $2Q$ goods in total, at
least $Q$ low goods must be assigned to charity.

We claim that after this prefix each agent must already hold at least
$Q-k$ low goods. Otherwise, let $S$ consist of $Q$ low goods currently
held by charity. Then, $S$ is feasible since $\size_i(S)=1$, and for every
$X\subseteq S$ with $|X|\le k$,
\[
\val(S\setminus X)\ge(Q-k)\delta,
\]
while the agent's value is strictly smaller than $(Q-k)\delta$, contradicting exact
FEF$ k$. Since the remaining predicted goods all have value $0$, later
arrivals cannot repair such a violation.

\paragraph{Actual instance.}
Now, keep the same prediction, the same horizon, and the same first
$3Q$ arrivals, but replace the final $Q+2k$ zero-padding goods by
\emph{unpredicted high} goods of value $1$ and size $1/Q$.

The zero-padding goods are missing (since they have value $0$, they contribute nothing to error), and the high goods are excess. The prediction error is therefore
\[
\Delta(f,\widehat f)=\sum_\theta\val_\theta|f_\theta-\widehat f_\theta| =1\cdot(Q+2k)=Q+2k
\]
since only the value-$1$ and value-$0$ types differ, and the latter
contributes nothing to the weighted error.

Since the two instances have identical low prefixes,
$\Alg$ makes exactly the same decisions during the first $3Q$
arrivals. Hence, each agent already holds at least $Q-k$ low goods and
has capacity for at most $k$ additional high goods, due to the capacity of size $Q$. Consequently, at most
$2k$ of the $Q+2k$ high-valued goods can be accepted, so at least $Q$
are assigned to charity.

Let $S$ be any $Q$ of these high-valued charity goods. Then,
$\size_i(S)=1$, and for every $X\subseteq S$ with $|X|\le k$,
\[
\val(S\setminus X)\ge Q-k.
\]
On the other hand, each agent receives value at most $Q\delta+k$,
consisting of at most $Q$ low goods of value $\delta$ together with at
most $k$ high-valued goods. Therefore,
\[
\viol_k(A)\ge(Q-k)-(Q\delta+k)
=Q-2k-Q\delta.
\]

Finally, choose a sufficiently small $\delta$ so that $Q\delta\le\frac{\zeta}{2}(Q+2k)$.
Then,
\begin{equation*}
    Q-2k-Q\delta \ge
(Q-2k)-\frac{\zeta}{2}(Q+2k) \ge
\left(1-\frac{\zeta}{2}\right)(Q+2k)
-\frac{\zeta}{2}(Q+2k) =
(1-\zeta)(Q+2k),
\end{equation*}
where the second inequality uses
\[
Q-2k>
\left(1-\frac{\zeta}{2}\right)(Q+2k).
\]
Since $\Delta(f,\widehat f)=Q+2k$, the claimed lower bound follows.
\end{proof}

\subsection{Proof of Theorem~\ref{thm:marginals}}
\label{app:proof:thm:marginals}
\MarginalsDoNotSuffice*

\begin{proof}
We construct two instances having identical predicted value and size
marginals, together with the same horizon, total value, maximum item
value, and arrival prefix. They differ only in the correlation between
values and sizes. The first instance forces every correct algorithm to
reserve almost all of each agent's capacity for low-valued goods,
whereas the second requires that capacity for high-valued goods. Since
a deterministic algorithm cannot distinguish the two instances before
the common prefix ends, it cannot be correct on both.

Let $\bud_1=\bud_2=1$. Choose $Q>2k$ and $\delta>0$ so that
$k+Q\delta<Q-k$, and let $H=Q+2k$. The predicted value multiset
consists of $3Q$ values $\delta$, $H$ values $1$, and $H$ values $0$.
The predicted size multiset consists of $3Q+H$ sizes $1/Q$ and $H$
sizes $2$. Both predicted multisets contain $3Q+2H$ entries. Hence, the
announced horizon is $T:=3Q+2H$.
The total value is $3Q\delta+H$, and the maximum item value is $1$.
The first $3Q$ arrivals are low-valued goods of size $1/Q$. As before,
at least $Q$ of these goods must be assigned to charity. We now
consider two continuations of this common prefix, both consistent with
the same predicted value and size marginals.

\paragraph{Continuation I (high value paired with infeasible size).}
The remaining goods consist of $H$ goods of value $1$ and size $2$,
together with $H$ goods of value $0$ and size $1/Q$. Since every
value-$1$ good has size $2>1$, it is infeasible for both agents and
can never appear in a feasible comparison set. The only feasible goods
remaining therefore have value $0$. As in the proof of
\Cref{thm:pred-lb}, exact FEF$ k$ requires each agent to hold at least
$Q-k$ low-valued goods after the common prefix. 
Otherwise, $Q$ low charity goods yield a violation as in \Cref{thm:pred-lb}.

\paragraph{Continuation II (high value paired with feasible size).}
Now, interchange the sizes of the remaining goods. The last arrivals are
$H$ goods of value $1$ and size $1/Q$, together with $H$ goods of value
$0$ and size $2$. Since the common prefix is unchanged, the algorithm
reaches exactly the same state as in Continuation~I. Thus, each agent
already holds at least $Q-k$ low-valued goods and has capacity for at
most $k$ additional feasible goods. Of the $H=Q+2k$ feasible
high-valued goods, at least $Q$ must therefore be assigned to charity.

Let $S$ consist of any $Q$ of these high-valued charity goods. Then,
$\size_i(S)=1$, while $\val(S\setminus X)\ge Q-k$ for every $X\subseteq S$ with $|X|\le k$. On the other hand, each
agent's value is at most $Q\delta+k<Q-k$, contradicting exact
FEF$ k$.

The two continuations have identical value marginals, size marginals,
horizon, total value, and maximum item value. They differ only in the
correlation between values and sizes. Thus, no deterministic online
algorithm receiving only these marginal predictions can distinguish the
two instances before their common prefix ends, even though correctness
requires different decisions. Therefore, predicting marginal
information alone is insufficient to guarantee exact FEF$ k$.
\end{proof}

\section{Lower Bound for Weak Scalar Advice}
\label{app:advice}\label{app:weak-advice}

The learning-augmented framework of \Cref{sec:la} assumes predictions
over the joint value-size type space. This appendix considers a much
weaker information model in which the algorithm receives only a single
scalar parameter describing the density range. The resulting lower
bound serves as a sanity check: knowing the exact minimum positive
density together with a valid upper bound on the density spread does
not improve the sharp spread-bounded frontier.

\begin{definition}[Minimum-density advice with a spread bound]
\label{def:min-density-advice}

We consider instances having at least one positively valued good that
is individually feasible under the original unit budget. The advice
supplied to the algorithm consists of:

\begin{enumerate}
    \item the exact minimum positive density
    \[
        \rho_{\min}
        :=
        \min\{\rho(g):v(g)>0,\ s(g)\le1\};
    \]
    \item a valid upper bound $\gamma\ge1$ on the ratio of maximum
    to minimum positive density among individually feasible goods, as in
    Definition~\ref{def:density-spread}.
\end{enumerate}
The value $\rho_{\min}$ is supplied in the original numerical scale
of the instance. No additional information is provided: in particular,
the algorithm is not told the realized maximum density, the horizon,
the multiplicities of the density levels, or the stopping prefix.
\end{definition}

The algorithm may use the numerical value of $\rho_{\min}$
arbitrarily; in particular, no scale-invariance assumption is imposed.

\begin{theorem}[Weak scalar advice does not improve the frontier]
\label{thm:min-density-advice-lb}

Fix $n\ge1,
\rho_{\min}>0,
\gamma>1,$
and let $\rho_{\max}:=\gamma\rho_{\min}.$
For every $\alpha\in
\left(
\frac1{1+\ln\gamma},
1
\right]$
and every deterministic online algorithm supplied with the exact
minimum positive density $\rho_{\min}$ together with the spread bound
$\gamma$, there exists, for every
$\sigma_{\max}>0$, a finite unknown-horizon instance with common
valuations and sizes common across agents such that

\begin{enumerate}
\item every good has size at most $\sigma_{\max}$;
\item the minimum and maximum positive densities are exactly
      $\rho_{\min}$ and $\rho_{\max}$, respectively;
\item the realized density spread is exactly $\gamma$; and
\item the algorithm does \emph{not} produce an
      $\alpha$-FEF1 allocation.
\end{enumerate}

\end{theorem}

\begin{proof}
We reduce the claimed lower bound to the frontier lower bound of
\Cref{thm:small-good-lower} by scaling every value by the advised
minimum density.

Fix a deterministic online algorithm $\mathcal A$ receiving the
advice $(\rho_{\min},\gamma)$. We define a deterministic online
algorithm $\widehat{\mathcal A}$ for normalized instances whose
positive densities lie in $[1,\gamma]$.

Whenever a normalized good $g$ with value $\widehat v(g)$ and size
$s(g)$ arrives, algorithm $\widehat{\mathcal A}$ forms a simulated
good having the same size and value $v(g):=\rho_{\min}\widehat v(g)$,
and presents this simulated arrival to $\mathcal A$, together with
the advice $(\rho_{\min},\gamma)$. Algorithm
$\widehat{\mathcal A}$ then assigns the normalized good to the same
recipient selected by $\mathcal A$ for the simulated good. This defines a deterministic online algorithm because the transformation uses only the current normalized arrival and the fixed constant $\rho_{\min}$. Since only the values are
rescaled, while sizes and budgets are unchanged, the simulation
preserves feasibility exactly when the corresponding normalized assignment is feasible.

Since $\alpha\in ( \frac1{1+\ln\gamma}, 1 ]$, applying \Cref{thm:small-good-lower} to
$\widehat{\mathcal A}$ with
$\varepsilon=0$ yields a fixed finite normalized instance
$\widehat{\mathcal H}$ such that

\begin{enumerate}[label=\textnormal{(\alph*)}]
\item every good has size at most $\sigma_{\max}$;
\item the minimum and maximum positive densities are exactly
      $1$ and $\gamma$, respectively; and
\item the allocation returned by $\widehat{\mathcal A}$ is \textit{not}
      $\alpha$-FEF1.
\end{enumerate}

Construct the corresponding original-scale instance
$\mathcal H$ by multiplying every value in
$\widehat{\mathcal H}$ by
$\rho_{\min}$ to obtain $v(g):=\rho_{\min}\widehat v(g)$, while leaving every size unchanged. Then, every positive density is multiplied by $\rho_{\min}$. Hence, $\min_{g:v(g)>0,\,s(g)\le1}\rho(g)=\rho_{\min}$ and $\max_{g:v(g)>0,\,s(g)\le1}\rho(g)
=\gamma\rho_{\min}
=\rho_{\max}$,
so the realized density spread is exactly $\gamma$.

By construction, given advice $\rho_{\min}$ and spread bound $\gamma$, the execution of $\mathcal A$ on
$\mathcal H$ is recipient-by-recipient identical to the execution of
$\widehat{\mathcal A}$ on
$\widehat{\mathcal H}$, since
$\widehat{\mathcal A}$ presents exactly the rescaled arrival sequence
to $\mathcal A$.

Finally, multiplying every value by the positive constant
$\rho_{\min}$ preserves every multiplicative FEF1 comparison:
\[
\widehat v(A_i)
<
\alpha\widehat v(S\setminus\{x\})
\iff
\rho_{\min}\widehat v(A_i)
<
\alpha\rho_{\min}\widehat v(S\setminus\{x\}).
\]
Hence the $\alpha$-FEF1 violation produced by
$\widehat{\mathcal A}$ on
$\widehat{\mathcal H}$ is preserved under scaling, yielding an
$\alpha$-FEF1 violation of $\mathcal A$ on
$\mathcal H$. Since $\mathcal H$ inherits the unknown-horizon
property and the same item sizes as
$\widehat{\mathcal H}$, it satisfies all the required properties.
\end{proof}

The simulation does not assume that $\mathcal A$ is scale invariant.
Instead, the auxiliary algorithm
$\widehat{\mathcal A}$ explicitly rescales each normalized arrival to
the original numerical scale before presenting it to
$\mathcal A$. The frontier lower bound is therefore applied to
$\widehat{\mathcal A}$, while the simulation transfers the resulting
hard instance back to the original advice model.

Although the constructed hard instance satisfies $\rho_{\max}=\gamma\rho_{\min}$, the algorithm is never given $\rho_{\max}$ explicitly. The advice
consists only of the exact minimum positive density
$\rho_{\min}$ together with the valid spread bound $\gamma$.

Theorem~\ref{thm:min-density-advice-lb} should be viewed as a
sanity-check lower bound for a particularly weak advice model. Through
the simulation reduction, the advice effectively specifies only the
lower endpoint of the normalized density interval, while revealing
nothing about the stopping prefix or the distribution of mass across
higher density levels. Consequently, the theorem does not constitute an
advice-complexity result and does not rule out richer forms of side
information, such as predicted quantiles, density histograms,
reservation profiles, or the joint value-size type counts used by
\textsc{Plan-Reserve-Fulfill}. Rather, it shows only that knowing the
minimum positive density together with a spread bound does not improve
the sharp spread-bounded frontier.

%====================================================================

\end{document}